\documentclass[preprint,12pt,authoryear]{elsarticle}

\usepackage{amssymb}

\usepackage{latexsym, amscd, amsxtra, amsmath,amsthm,mathtools }
\RequirePackage[colorlinks,citecolor=blue,urlcolor=blue]{hyperref}
\usepackage{graphics, graphicx, color}
\usepackage{ifpdf}
\usepackage{multirow}
\graphicspath{{figures/}{./figures/}}

\journal{CSDA}

\newtheorem{thm}{Theorem}

\newtheorem{lem}[thm]{Lemma}
\newtheorem{prop}[thm]{Proposition}
\theoremstyle{definition}

\newtheorem*{defn*}{Definition}
\theoremstyle{remark}

\newtheorem{example}{Example}

\newcommand{\norm}[1]{\|#1\|}

\newcommand{\Real}{\mathbb R}
\newcommand{\eps}{\varepsilon}

\def\arg{\mbox{arg}}

\def\1v{{\bf 1}}
\def\0v{{\bf 0}}

\begin{document}

\begin{frontmatter}



\title{Continuum directions for supervised dimension reduction}


\author{Sungkyu Jung}

\address{Department of Statistics, University of Pittsburgh,
Pittsburgh, PA 15260, U.S.A.}
\ead{sungkyu@pitt.edu}
\begin{abstract}
Dimension reduction of multivariate data supervised by auxiliary information is considered. A series of basis for dimension reduction is obtained as minimizers of a novel criterion. The proposed method is akin to continuum regression, and the resulting basis is called continuum directions. With a presence of binary supervision data, these directions continuously bridge the principal component, mean difference and linear discriminant directions, thus ranging from unsupervised to fully supervised dimension reduction. High-dimensional asymptotic studies of continuum directions for binary supervision reveal several interesting facts. The  conditions under which the sample continuum directions are inconsistent, but their classification performance is good, are specified. While the proposed method can be directly used for binary and multi-category classification, its generalizations to incorporate any form of auxiliary data are also presented. The proposed method enjoys fast computation, and the performance is better or on par with more computer-intensive alternatives.
\end{abstract}

\begin{keyword}
continuum regression \sep dimension reduction \sep linear discriminant analysis \sep high-dimension, low-sample-size (HDLSS) \sep maximum data piling \sep principal component analysis

\MSC 60K35

\end{keyword}

\end{frontmatter}



\section{Introduction}

In modern complex data, it becomes increasingly common that multiple data sets are available. We consider the data situation where a supervised dimension reduction is naturally considered. Two types of data are collected on a same set of subjects: a data set of primary interest $X$ and an auxiliary data set $Y$.
The goal of supervised dimension reduction is to delineate major signals in $X$, dependent to $Y$. Relevant application areas  include genomics (genetic studies collect both gene expression and SNP data---\cite{li2016supervised}), finance data (stocks as $X$ in relation to characteristics $Y$ of each stock: size, value, momentum and volatility---\cite{connor2012efficient}), and batch effect adjustments \citep{lee2014covariance}.

There has been a number of work in dealing with the multi-source data situation.
\cite{lock2013joint} developed JIVE to separate joint variation from individual variations. Large-scale correlation studies can identify millions of pairwise associations between two data sets via multiple canonical correlation analysis \citep{witten2009extensions}. These methods, however, do not provide supervised dimension reduction of a particular data set $X$, since all data sets assume an equal role.

In contrast, reduced-rank regression  \cite[RRR,][]{izenman1975reduced,tso1981reduced} and envelop models \citep{cook2010envelope} provide {sufficient} dimension reduction \citep{cook2005sufficient} for regression problems.  See \cite{cook2013envelopes} for connections between envelops and partial least square regression.
Variants of principal component analysis (PCA) have been proposed to incorporate auxiliary information; see \cite{fan2016projected} and references therein.  Recently, \cite{li2016supervised} proposed SupSVD, a supervised PCA that encompasses regular PCA to RRR. Our goal is similar to that of SupSVD, which extends RRR and envelop models, in that the primary and auxiliary data sets play different roles.
We consider a basis (or subspace) recovery to extract the part of main data set which is relevant to the auxiliary data set. Unlike SupSVD, which provides a fully supervised dimension reduction, we seek a unified framework that covers a wide spectrum from fully-supervised to unsupervised dimension reduction.

A potential drawback of fully supervised dimension reduction as a preprocessing for further application of predictive modeling is a \emph{double-dipping} problem: The same signal is considered both at dimension reduction and at classifiers.
In  high dimensional data situations, small signals can sway the whole analysis, often leading to a spurious finding that can not be replicated in subsequent studies. A regularized semi-supervised dimension reduction has a potential to mitigate the double-dipping problem.

We propose a semi-supervised basis learning for the primary data that covers a wide range of spectrum from supervised to unsupervised dimension reduction. A meta-parameter $\gamma \in [0,\infty)$ is introduced to control the degrees of supervision. The spectrum of dimension reduction given by different $\gamma$ is best understood when there exists a single binary supervision. In such a special case, the directional vectors of the basis continuously bridge the principal component direction, mean difference and Fisher's linear discriminant directions.

The proposed method was motivated by the continuum regression \citep{Stone1990}, regressors ranging from the ordinary least square to the principal component regression. In the context of regression, our primary data set is predictors while the auxiliary data are the response. The new basis proposed in this work, called \textit{continuum directions},  can be used with multivariate supervision data, consisting of either categorical or continuous variables.

We also pay a close attention to the high-dimension, low-sample-size situations (or the $p \gg n$ case), and give a new insight on the
maximum data piling (MDP) direction $w_{MDP}$, proposed as a discriminant direction for binary classification by \cite{Ahn2010}. In particular, we show that $w_{MDP}$ is a special case of the proposed continuum direction, and if $p \gg n$, MDP is  preferable to linear discriminant directions in terms of Fisher's original  criterion for linear discriminant analysis \citep[LDA,][]{Fisher1936}. We further show, under the high-dimension, low-sample-size asymptotic scenario \citep{Hall2005}, although the empirical continuum directions are inconsistent with their population counterparts, the classification performance using the empirical continuum directions can be good, if the signal strength is large enough.

As an application of the continuum directions, we endeavor to use the continuum directions in classification problems.
Recently, numerous efforts to improve classifications for the $p \gg n$ situation have been made. Linear classifiers such as LDA, the support vector machine \citep{vapnik2013nature} or distance weighted discrimination \citep{Marron2007,Qiao2009} often yield better classification than nonlinear methods, in high dimensional data analysis. A recent trend is sparse estimations. \cite{Bickel2004} studied the independence rule, ignoring off-diagonal entries of $S_W$. Additionally assuming sparsity of the population mean difference, \cite{Fan2008} proposed the features annealed independence rule (FAIR).
\cite{Wu2009} and \cite{Shao2011} proposed sparse LDA estimations, and
\cite{Clemmensen:2011} proposed sparse discriminant analysis (SDA) to learn sparse basis for multi-category classification.
\cite{Cai2011} proposed the linear programming discriminant rule (LPD) for sparse estimation of the discriminant direction vector. The sparse LDA, SDA and LPD are designed to work well if their sparsity assumptions are satisfied. Sophisticated methods such as those of \cite{Wu2009} and \cite{Cai2011} usually suffer from heavy computational cost.
Our method, when applied to the binary classification problem, leads to  analytic solutions, and the computation times are scalable. We show via simulation studies that classification performance using the continuum directions is among the best when the true signal is not sparse and  the variables are highly correlated.

The rest of the paper is organized as follows. In Section~\ref{sec:CDD}, we introduce continuum directions and discuss its relation to continuum regression. In the same section, we provide some insights for continuum directions in high dimensions. In Section~\ref{sec:computation}, we show numerical procedures that are efficient for high-dimensional data. Simulation studies for classification performance in high dimensions can be found in Section~\ref{sec:sim}. We further show advantages of our method by a few real data examples in Section~\ref{sec:realdata}. We conclude with a discussion. Proofs are contained in Appendix.

\section{Continuum directions}\label{sec:CDD}

\subsection{Motivation} \label{sec:CDAbinary}

To motivate the proposed framework for dimension reduction, we first analyze a special case where the supervision data consist of a binary variable. We discuss a few meaningful directions for such situations, viewed in terms of a two-group classification problem. These directions are special cases of the continuum directions, defined later in (\ref{eq:CDAcriterion}).

Let $n_1$ and $n_2$ be the numbers of observations in each group and $n = n_1+n_2$. Denote $\{x_{11},\ldots,x_{1n_1}\}$ and $\{x_{21},\ldots,x_{2n_2}\}$ for the $p$-dimensional observations of the first and second group, respectively. In our study it is sufficient to keep the sample variance-covariances. Denote $S_W = \frac{1}{n}(\sum_{i=1}^{n_1}(x_{1i} - \bar{x}_{1\cdot})(x_{1i} - \bar{x}_{1\cdot})^T + \sum_{i=1}^{n_2}(x_{2i} - \bar{x}_{2\cdot})(x_{2i} - \bar{x}_{2\cdot})^T)$ for the within-group variance matrix, i.e. the estimated (pooled) common covariance, and $S_B = \frac{n_1n_2}{(n_1+n_2)^2} (\bar{x}_{1\cdot} - \bar{x}_{2\cdot})(\bar{x}_{1\cdot} - \bar{x}_{2\cdot})^T$ for the between-group variance matrix. The total variance matrix is $S_T = \frac{1}{n}(\sum_{i=1}^{n_1}(x_{1i} - \hat{\mu})((x_{1i} - \hat{\mu})^T + \sum_{i=1}^{n_2}(x_{2i} - \hat{\mu})(x_{2i} - \hat{\mu})^T)$ with the common mean $\hat{\mu} = (n_1 \bar{x}_{1\cdot} + n_2\bar{x}_{2\cdot})/n$, and $S_T = S_W + S_B$.

Fisher's criterion for discriminant directions is to find a direction vector $w$ such that, when data are projected onto $w$, the between-variance $w^TS_Bw$ is maximized while the within-variance $w^TS_Ww$ is minimized. That is, one wishes to find a maximum of
\begin{equation}\label{eq:Fishercriterion}
T(w) =  \frac{w^TS_Bw}{w^TS_Ww}.
\end{equation}
If $S_W$ is non-singular, i.e. the data are not collinear and $p \le n-2$, the solution is given by $w_{LDA} \propto S_W^{-1}d$, where $d = \bar{x}_{1\cdot} - \bar{x}_{2\cdot}$. It has been a common practice to extend the solution to the case $p > n-2$ using a generalized inverse, i.e.,
 $$w_{LDA} \propto S_W^{-}d,$$
 where $A^-$ stands for the Moore--Penrose pseudoinverse of square matrix $A$. 

In retrospect, when $\mbox{rank}(S_W) < p $, Fisher's criterion is ill-posed since there are infinitely many $w$'s satisfying $w^TS_Ww = 0$. Any such $w$, which also satisfies $w^TS_Bw > 0 $, leads to $T(w) = \infty$.
In fact, in such a situation, $w_{LDA}$ is not a maximizer of $T$ but merely a critical point of $T$.  
\cite{Ahn2010} proposed a maximal data piling (MDP) direction $w_{MDP}$ which maximizes the between-group variance $w^TS_Bw$ subject to $w^TS_Ww = 0$, and is
$$w_{MDP} \propto S_T^{-}d.$$
Note that $w_{MDP}$ also maximizes a criterion
\begin{equation}\label{eq:MDPcriterion}
T_{MDP}(w) = \frac{w^TS_Bw}{w^TS_Tw}.
\end{equation}
In the conventional case where $n \ge p$, the criteria (\ref{eq:Fishercriterion}) and (\ref{eq:MDPcriterion}) are equivalent up to a constant, and $w_{MDP} = w_{LDA}$. We discuss further in Section~\ref{sec:FisherMDP-LDA} that MDP is more preferable than LDA in the high-dimensional situations.

A widely used modification to Fisher's criterion is to shrink $S_W$ toward a diagonal matrix, leading to
\begin{equation}\label{eq:RRcriterion}
T_\alpha(w) =  \frac{w^TS_Bw}{w^T(S_W+ \alpha I)w}, \ \mbox{ for some } \alpha \ge 0.
\end{equation}
This approach has been understood in a similar flavor to ridge regression \citep{Hastie2009}. The solution of the above criterion is simply given by
$w^{R}_{\alpha} \propto (S_W+\alpha I)^{-1} d.$ A special case is in the limit $\alpha \to \infty$, where the solution $w^{R}_\infty$ becomes the direction of mean difference (MD) $w_{MD} \propto d$, which maximizes
\begin{equation}\label{eq:MDcriterion}
T_{MD}(w) = {w^TS_Bw},
\end{equation}
with a conventional constraint $w^Tw = 1$.


In high dimensional data situations, utilizing the principal components is a natural and nonparametric way to filter out  redundant noise. Principal component analysis (PCA) reduces the dimension $p$ to some low number $p_0$ so that the subspace formed by the first $p_0$ principal component directions contains maximal variation of the data among all other $p_0$-dimensional subspaces. In particular, the first principal component direction $w_{PC1} $ maximizes the criterion for the first principal component direction,
\begin{equation}\label{eq:PCAcriterion}
T_{PCA}(w) =\frac{w^TS_Tw}{w^Tw}.
\end{equation}

The important three directions of MDP, MD and PCA differ only in criteria maximized. With the constraint $w^Tw = 1$, the criteria (\ref{eq:MDPcriterion})--(\ref{eq:PCAcriterion}) are functions of total-variance $w^TS_Tw$ and between-variance $w^TS_Bw$. For the binary supervision case, a generalized criterion that embraces all three methods is
\begin{equation}\label{eq:CDAcriterion}
T_\gamma (w) =  (w^TS_Bw)(w^TS_Tw)^{\gamma - 1} \ \mbox{ subject to } \ w^Tw = 1,
\end{equation}
where $\gamma$ takes some value in $[0,\infty)$. The special cases are MDP as $\gamma \to 0$, MD at $\gamma = 1$, and PCA when $\gamma \to \infty$. The direction vector $w_\gamma$ that maximizes $T_\gamma$ is called the \emph{continuum direction} for $\gamma$.

\subsection{General Continuum directions}\label{sec:CDA_special}
The continuum direction (\ref{eq:CDAcriterion}) defined for the binary supervision is now generalized to incorporate any form of supervision.

Denote $X = [x_1,\ldots,x_n]$ for the $p\times n$ primary data matrix and $Y$ for the $r \times n$ matrix with secondary information. The matrix $Y$ contains the supervision information that can be binary, categorical, and continuous.
For example if the supervision information is a binary indicator for two-group classification with group sizes $n_1$ and $n_2$, as in Section~\ref{sec:CDAbinary}, then the matrix $Y$ can be coded as the $2 \times n $ matrix $Y^T = [ n_1(e_{1} - j_n) ; n_2(e_{2} - j_n)]$ where $j_n = n^{-1}(1,1,\ldots,1)^T = n^{-1}1_n$ and $e_k$ is the length-$n$ vector, whose $i$th element is $n_k^{-1}$ if the $i$th subject is in the $k$th group, and zero otherwise. Similarly, if the supervision information is multicategory with $K$ groups, then $Y$ is the $K \times n$ matrix whose $k$th row is $n_k(e_k - j_n)^T$, where $n_k$ is the number of observations belonging to category $k$. If the supervision is continuous and multivariate, such as responses in multivariate regression, then the matrix $Y$ would collect centered measurements of response variables.

Assuming for simplicity that $X$ is centered, we write the total variance-covariance matrix of $X$ by $S_T = n^{-1}XX^T$, and the $Y$-relevant variance-covariance matrix of $X$ by $S_B = n^{-1} (XY^T)(XY^T)^T$.
A completely \textit{unsupervised} dimension reduction can be obtained by eigendecomposition of $S_T$. On the contrary, a fully-supervised approach is to focus on the column space of $S_B$, corresponding to the mean difference direction when $Y$ is binary. An extreme approach that nullifies the variation in $X$ to maximize the signals in $Y$ can be obtained by eigendecomposition of $S_T^- S_B$. When $Y$ is categorical, this reduces to the reduced-rank LDA.

Generalizing (\ref{eq:CDAcriterion}), the following approach encompasses the whole spectrum from the supervised to unsupervised dimension reduction. A meta-parameter $\gamma \in [0, \infty)$ controls the degree of supervision.
For each $\gamma$, we obtain a basis $\{w_{(1)},\ldots, w_{(\kappa)}\}$ for dimension reduction of $X$ in a sequential fashion. In particular, given $w_{(1)},\ldots,w_{(k)}$, the
 $(k+1)$th direction is defined by $w$ maximizing \begin{align}\label{eq:CDA}
T_\gamma(w) &= (w^TS_Bw) (w^T S_T w)^{\gamma - 1},\\
\mbox{subject to}\ &\ w^T w = 1\ \mbox{ and } w^T S_T w_{(\ell)} = 0, \ \ell = 1,\ldots, k. \nonumber
\end{align}
%
The sequence of  directions $\{w_{(\ell)}: \ell = 1,\ldots,\kappa\}$ for a given value of $\gamma$ is then $S_T$-orthonormal to each other: $w_{(\ell)}^T w_{(\ell)} = 1, w_{(\ell)}^T S_T w_{(l)} = 0$ for $\ell \neq l$. An advantage of requiring $S_T$-orthogonality is that the resulting scores $z_{\ell,i} = x_i^T w_{(\ell)}$ are uncorrelated with $z_{l,i}$ for $\ell \neq l$. This is desirable if these scores are used for further analysis, such as a classification based on these scores.

In sequentially solving (\ref{eq:CDA}), choosing large $\gamma$ provides nearly unsupervised solutions while $\gamma\approx 0$  yields an extremely supervised dimension reduction.
The spectrum from unsupervised to supervised dimension reduction is illustrated in a real data example shown in Example~\ref{ex:lung cancer}.

\begin{example}\label{ex:lung cancer}
We demonstrate the proposed method of dimension reduction for  a real data set from a microarray study. This data set, described in detail in  \cite{Bhattacharjee2001}, contains
$p = 2530$ genes (primary data) from $n= 56$ patients  while the patients are labeled by four different lung cancer subtypes (supervision data). The primary dataset $X$ is the $p \times n$ matrix of normalized gene expressions, while the supervision data is the $4 \times n$ matrix  $Y$, coded to use the categorical cancer subtypes as the supervision.

The continuum directions can provide basis of dimension reduction, ranging from the unsupervised ($\gamma\approx \infty$) to the fully supervised ($\gamma\approx 0$). In Fig. \ref{fig:NeilHayesCDAm2}, the projected scores of the original data are plotted for four choices of $\gamma$.

\begin{figure}[tb]
\centering
\ifpdf
\vskip -0.2in
    \includegraphics[width=0.9\textwidth]{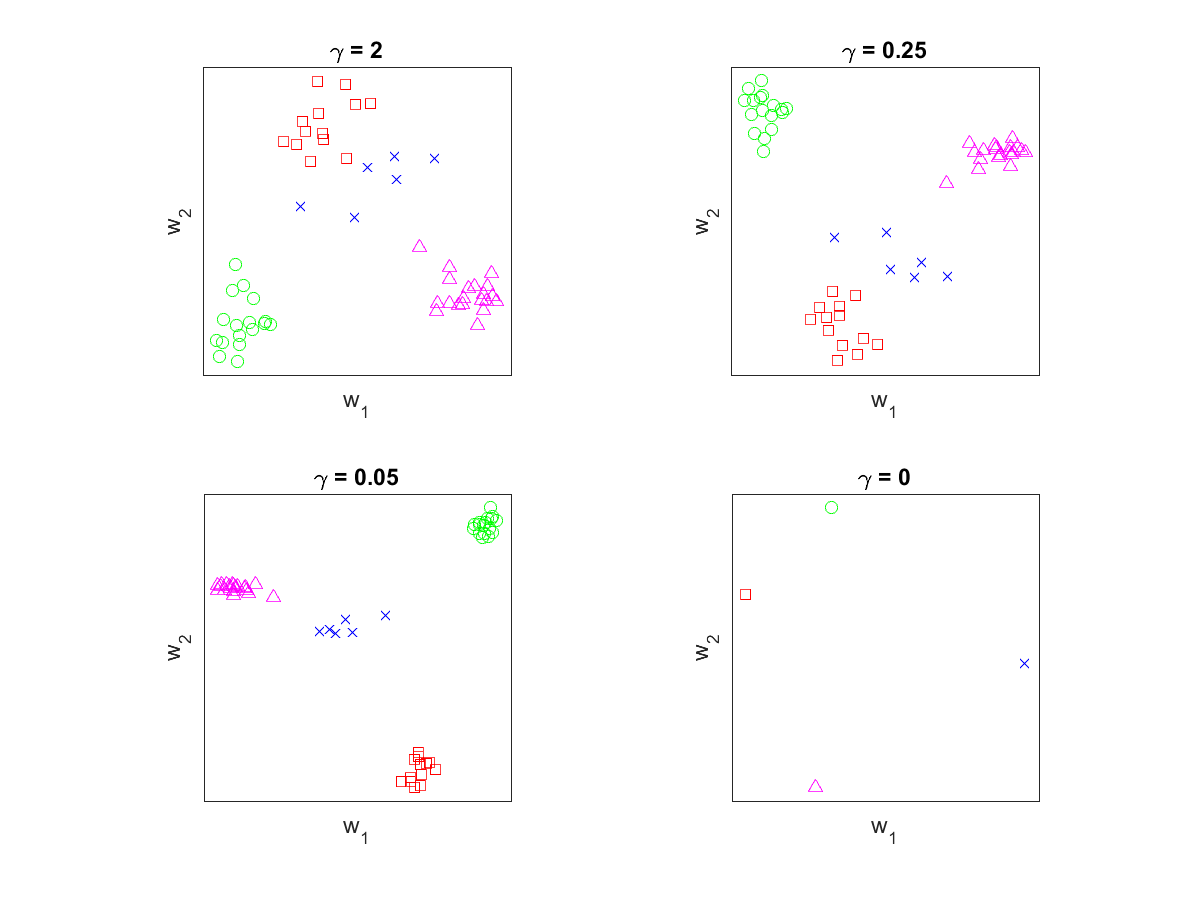}
\vskip -0.2in
\else
    \includegraphics[width=0.9\textwidth]{NeilHayesCDAm2.eps}
\fi
\caption{Spectrums of supervised dimension reduction for the data set of \cite{Bhattacharjee2001}. Shown are the projection scores to the first two continuum  directions, for various values of $\gamma$.}
\label{fig:NeilHayesCDAm2}
\end{figure}

A dimension reduction by PCA has been useful for this data set, since the four subtypes are visually separated by using the first few sample principal components. The principal component scores are similar to those plotted in the first panel of Fig.  \ref{fig:NeilHayesCDAm2} when $\gamma$ is large enough. On the other hand, a fully supervised dimension reduction given by the MDP directions,  plotted in the bottom right panel, nullifies any variation in the primary data set. Specifically, all observations corresponding to the same subtype project to a single point, a feature due to the high dimensionality. Thus the projected scores for $\gamma =0$ contain  information only relevant to the supervision.

The continuum direction as a function of $\gamma$ is continuous (shown later in Proposition~\ref{prop:RRtoMDP}), thus the projected scores are also continuous with respect to $\gamma$. The continuous transition of the scores from large $\gamma$ to small $\gamma$ in Fig.~\ref{fig:NeilHayesCDAm2} is thus expected.
The question of which value of $\gamma$ to use in final dimension reduction depends on the purpose of analysis.
For exploratory analysis, several values of $\gamma$ may be used to examine the data from  a variety of viewpoints.  If the dimension reduction is performed for regression or classification, a cross-validation can be used, which is discussed in Section~\ref{sec:classification}.

\end{example}

\subsection{Continuum directions for classification}\label{sec:classification}

When the supervision data is binary or categorical, it is natural to seek an application of continuum directions for the basis of classification. In particular, for the binary supervision case, as shown in Section~\ref{sec:CDAbinary}, the continuum direction $w_\gamma$ can be thought of as the normal direction to the separating hyperplane.

In the general $K$-group situation, for each $\gamma>0$, the sequence of directions $\{w_{(\ell)}: \ell = 1,\ldots, \kappa \}$ are used to obtain dimension-reduced scores $z_{\ell,i} = x_i^T w_{(\ell)}$, $\ell = 1,\ldots, \kappa$, for secondary discriminant analysis. In particular, we choose $\kappa = K-1$ and use $[z_{1,i},\ldots,z_{\kappa,n}]$, $i = 1,\ldots, n$, in training the ordinary LDA.
For a new observation $x_\ast$, the scores $z_{(\ell,\ast)} = x_\ast^T w_{(\ell)}$ are used for the prediction by the trained LDA.
This classification rule is called continuum discriminant analysis (CDA).

The CDA depends on the choice of $\gamma$. A 10-fold cross-validation to minimize the expected risk with the 0-1 loss can be used to tune $\gamma$. We use a cross-validation index $CV(\gamma)$ that counts the number of misclassified observations for each given $\gamma$, divided by the total number of training sample. As exemplified with real data examples in Section~\ref{sec:realdata}, the index $CV(\gamma)$ is typically U-shaped. This is because that the two ends of the spectrum are quite extreme. Choosing $\gamma =0$ results in the unmodified LDA or MDP, while choosing $\gamma \approx \infty$ results in using PC1 direction for classification. In our real data examples, the minimizer of $CV(\gamma)$ is found in the interval $[0.2, 2.19]$.

%
%

\subsection{Relation to continuum regression}\label{sec:ridge and CR}

A special case of the proposed method, specifically (\ref{eq:CDAcriterion}) for the binary supervision, can be viewed as a special case of continuum regression \citep{Stone1990}.
The continuum regression leads to a series of regressors that bridges ordinary least squares, partial least squares and principal component regressions. In connection with the continuum directions for binary classification, ordinary least squares regression corresponds to LDA (or MDP in (\ref{eq:MDPcriterion})), and partial least squares corresponds to mean difference. In particular, in the traditional case where $n > p$, it is well known that $w_{LDA}$ is identical to the vector of coefficients of least squares regression, up to some constant.

Some related work has shed light on the relationship between continuum regression and ridge regression \citep{Sundberg1993,DeJong1994,Bjorkstrom1999}. A similar relationship can be established for our case when $S_B$ is of rank 1. For simplicity, we assume that the column space of $S_B$ is spanned by the vector $d$.  (In the binary classification case, $d = \bar{x}_{1\cdot} - \bar{x}_{2\cdot}$, as discussed in Section~\ref{sec:CDAbinary}.)  To find the continuum direction $w_\gamma$ that maximizes $T_\gamma (w)$ in  (\ref{eq:CDAcriterion}), differentiating the Lagrangian function $\log T_\gamma (w) - \lambda (w^Tw - 1)$  with respect to $w$ leads to the equation
\begin{equation}\label{eq:Lagrangian}
 \frac{S_B w}{w^TS_Bw} + (\gamma - 1) \frac{S_T w}{w^TS_Tw} - \lambda w = 0.
 \end{equation}
Left multiplication of $w^T$ leads to $\lambda = \gamma$. A critical point of the preceding equation system gives the maximum of $T_\gamma$. Since $\frac{S_B w}{w^TS_Bw} = \frac{d d^T w}{w^Td d^Tw} = \frac{1}{d^Tw}d$, one can further simplify the equation for a critical point
\begin{equation}\label{eq:ridgesolution}
w \propto (S_T + \frac{\gamma}{1-\gamma}\frac{w^TS_Tw}{w^Tw} I_p)^{-} d = (S_T + \alpha I_p)^{-} d := w^{R}_\alpha.
\end{equation}
For each $\gamma \in [0,1)$, there exists an $\alpha = \alpha(\gamma) \ge 0$ such that the continuum discriminant direction $w_\gamma$ is given by the ridge estimator $w^{R}_\alpha$. This parallels the observation made by \cite{Sundberg1993} in regression context. We allow negative $\alpha$, so that the relation to ridge estimators is extended for $\gamma>1$.

\begin{thm}\label{prop:ridgesolutionthm}
If $d$ is not orthogonal to all eigenvectors corresponding to the largest eigenvalue $\lambda_1$ of $S_T$,
then for each $\gamma>0$ there exists a number $\alpha \in (-\infty, -\lambda_1) \cup [0,\infty)$  such that $w_\gamma \propto (S_T + \alpha I)^{-} d$, including the limiting cases $\alpha \to 0, \alpha \to \pm\infty$ and $\alpha \to -\lambda_1$.
\end{thm}

The above theorem can be shown by an application of Proposition~2.1 of \cite{Bjorkstrom1999} who showed that, in our notation, the solution of $\max_{w} T_\gamma(w)$ is of the ridge form. See Appendix for a  proof of the theorem.

The relation between $\alpha$ and $\gamma$ is nonlinear and depends on $S_T$. A typical form of relation is plotted in Fig.~\ref{fig:IRIS}, and is explained in the following example.

\begin{example}
From Fisher's iris data, we chose `versicolor' and `virginica' as two groups each with $50$ samples. For presentational purpose, we use the first two principal component scores of the data. For a dense set of $\gamma \in [0, \infty)$, the corresponding $\alpha$ is plotted (in the left panel of Fig.~\ref{fig:IRIS}), which exhibits the typical relationship between $\gamma$ and $\alpha$.
The MDP at $\gamma = 0$ corresponds to the ridge solution with $\alpha = 0$. As $\gamma$ approaches 1, the corresponding ridge solution is obtained with $\alpha \to \pm \infty$. For $\gamma > 1$, $\alpha$ is negative and approaches $-\lambda_1$ as $\gamma \to \infty$.
 The continuum directions $\{w_\gamma:\gamma \in  [0,\infty)\}$ range from $w_{LDA}$ (which is the same as $w_{MDP}$ since $n > p$) to $w_{PCA}$ as illustrated in the right panel of Fig.~\ref{fig:IRIS}.
\begin{figure}[tb!]
\centering
\ifpdf
\vskip -2.5in
    \includegraphics[width=1\textwidth]{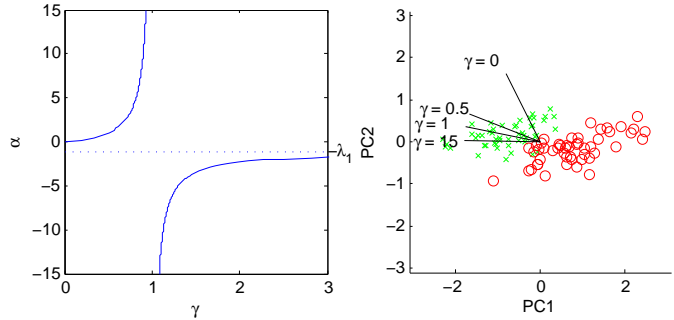}
\vskip -2.5in
\else
    \includegraphics[width=0.7\textwidth]{IRIScda.eps}
\fi
\caption{(left) Relation between $\gamma$ and $\alpha$, illustrated for the iris data. (right) Continuum directions $w_\gamma$ are overlaid on the scatter plot of the first two principal components. Different symbols represent different groups.}
\label{fig:IRIS}
\end{figure}
\end{example}

The ridge solution may not give a global maximum of $T_\gamma$ when the assumption in Theorem~\ref{prop:ridgesolutionthm} does not hold. An analytic solution for such a case is also provided in Proposition~\ref{prop:ridgesolutionexceptionthm} in Appendix.

\subsection{Continuum directions in high dimensions}\label{sec:FisherMDP-LDA}

In high-dimensional situations where the dimension $p$ of the primary data is much higher than the sample size $n$, the continuum directions are still well-defined.
We return to discuss that, if $p > n$, MDP has more preferable properties than LDA for binary classification.  The ridge solution plays an important role in the following discussion.

In the conventional case where $p \le n-2$, It is easy to see that the ridge criterion (\ref{eq:RRcriterion}) and its solution $w^{R}_{\alpha}$ (\ref{eq:ridgesolution}) bridge LDA and MD. However, if $p > n$ and thus $S_W$ is rank deficient, one extreme of the ridge criterion is connected to MDP but not to LDA.
The following proposition shows that $w^{R}_{\alpha}$ ranges from MD to MDP, giving a reason to favor MDP over LDA in high dimensions.

\begin{prop}\label{prop:RRtoMDP}
For $\alpha>0$, $w^{R}_{\alpha} \propto (S_T+\alpha I)^{-1} d$. Moreover $w^{R}_{\alpha}$ is continuous with respect to $\alpha \in (0,\infty)$. The boundaries meet MDP and MD directions, that is, $\lim_{\alpha \to 0} w^{R}_{\alpha} = w_{MDP}$ and $\lim_{\alpha \to \infty} w^{R}_{\alpha} = w_{MD}$.
\end{prop}

While $w_{MDP}$ is a limit of ridge solutions, $w_{LDA}$ does not meet with $w^{R}_{\alpha}$.
When $p > n$, $w_{MDP}$ is orthogonal to $w_{LDA}$ if the mean difference $d$ is not in the range of $S_W$, i.e., $\mbox{rank}(S_W) < \mbox{rank}(S_T)$ \citep{Ahn2010}. This fact and Proposition~\ref{prop:RRtoMDP} give $\lim_{\alpha \to 0} \mbox{angle}(w_{LDA}, w^{R}_{\alpha}) = 90^\circ$.

Algebraically, the discontinuity of the ridge direction  to $w_{LDA}$ comes from the discontinuity of the pseudoinverse. Heuristically, the discontinuity comes from the fact that $d$ does not completely lie in the column space  of $S_W$. In such a case, there is  a  direction vector $w_0$ orthogonal to the  column space of $S_W$ containing information about $d$ (i.e., $d^T w_0 \neq 0$). Using $S_W^-$ in LDA ignores such information. On the other hand, MDP uses $S_T^-$, which preserves all information contained in the special direction $w_0$.

The values of Fisher's criterion for various choices of $w$ in Fig.~\ref{fig:FisherscriterionNeilHayes} exemplify that $w_{MDP}$ should be used as Fisher discriminant direction rather than $w_{LDA}$ in high dimensions. In our experiments on classification (in Sections~\ref{sec:sim} and \ref{sec:realdata}), we check that the empirical performance of LDA is among the worst.

\begin{figure}[tb!]
\centering
\ifpdf
    \includegraphics[width=1\textwidth]{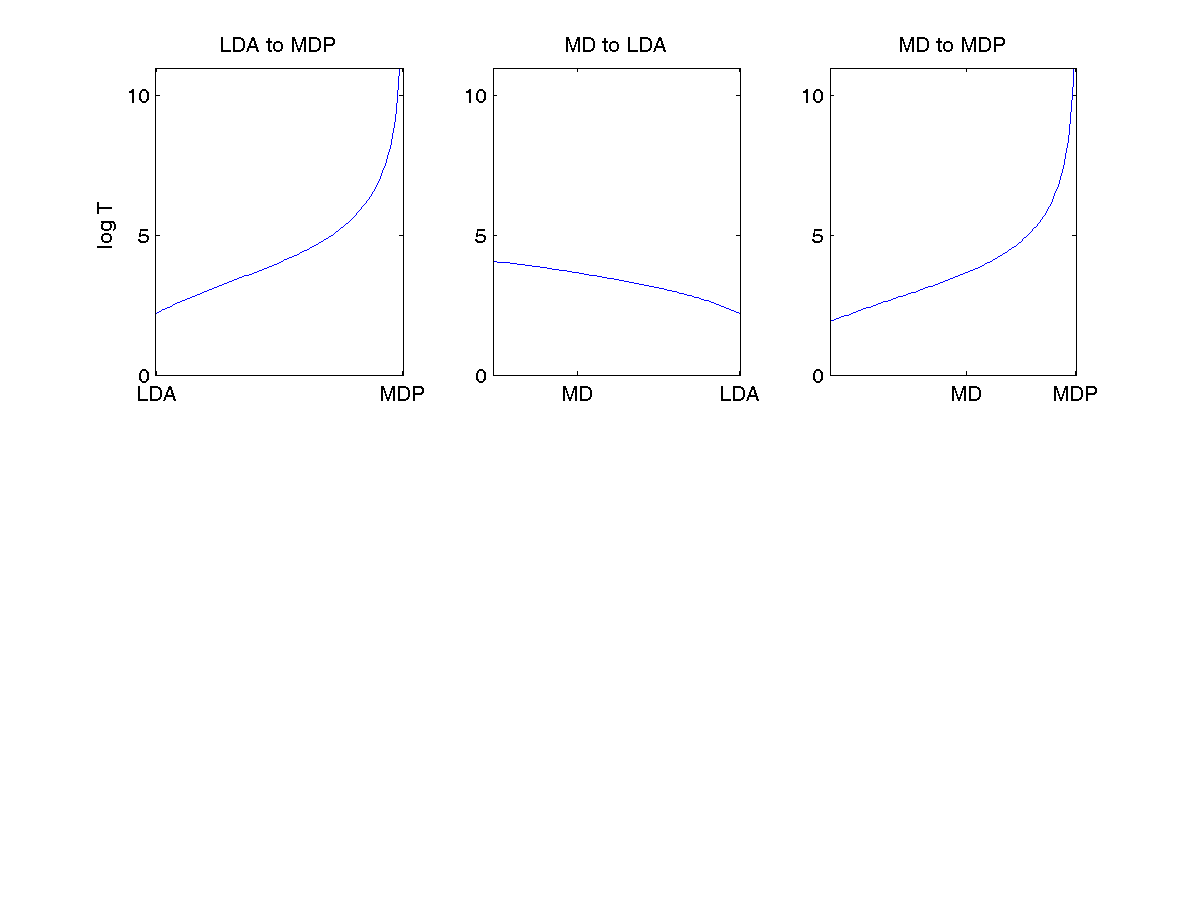}
\vskip -2in
\else
    \includegraphics[width=0.7\textwidth]{FisherscriterionNeilHayes.ps}
\fi
\caption{Fisher's $T(w)$ for directions discriminating two groups ($n_1=20, n_2 = 17$) in a microarray dataset with $p = 2530$ \citep{Bhattacharjee2001}. The three horizontal axes represent discriminant direction $w$ along the edges of the triangle formed by $w_{LDA},w_{MD}$, and $w_{MDP}$. LDA is not maximizing Fisher's criterion and is inferior to the mean difference, while $T(w_{MDP}) = \infty$.}
\label{fig:FisherscriterionNeilHayes}
\end{figure}

Our discussion so far assumes that the covariance matrices $S_T, S_W, S_B$ are the sample covariance matrices. It is well-known that these matrices are inconsistent estimators of the population covariance matrices when $p \gg n$, as $n \to \infty$.  Only with strong assumptions on the covariance and mean difference (such as sparsity), it is possible to devise consistent estimators.
%
%
In such situation, the sufficient statistics $S_T$ and $S_B$ can be replaced by consistent estimators $\widehat{\Sigma}_T$ and $\widehat{\Sigma}_B$, in the evaluation of the continuum directions (\ref{eq:CDA}). This approach has a potential to provide an estimator of $w_\gamma$, consistent with a suitably defined population continuum directions, when $p/n \to \infty$.
In the next section, we present a high-dimensional asymptotic study when $S_T$ and $S_B$ are used in computing the empirical continuum directions.

\subsection{HDLSS asymptotic study of continuum directions}\label{sec:high-d}

%
We employ the high-dimension, low-sample-size (HDLSS) asymptotics, that is, the asymptotic study of $p\to\infty$ while the sample size $n$ is held fixed, to understand the high-dimensional behaviors of the true and sample continuum directions. The HDLSS asymptotics has been successfully used in revealing the properties of conventional multivariate methods in high dimensions, such as classification \citep{Hall2005,Qiao2009}, PCA \citep{Jung2009a,yata2009a,zhou2015high}, and clustering \citep{Ahn2012}, to name a few. For a review of recent developments, see \cite{AOSHIMA2017}.

To set up, suppose that $x_{11},\ldots, x_{1n_1}$ are i.i.d. $N_p(\mu_1,\Sigma_1)$ and $x_{21},\ldots, x_{2n_2}$ are i.i.d. $N_p(\mu_2,\Sigma_2)$.
The empirical continuum directions $w_\gamma$ are given by (\ref{eq:CDAcriterion}) where $S_B$ and $S_T$ as defined in Section~\ref{sec:CDAbinary}.
By Theorem~\ref{prop:ridgesolutionthm}, the elements in the set of true continuum directions $\{w_\gamma : \gamma >0\}$ can also be parametrized by
\begin{equation}\label{eq:alpha-st}
\alpha(\gamma,S_T) = \frac{\gamma }{1-\gamma}\frac{ \omega_\gamma^T S_T \omega_\gamma }{\omega_\gamma^T  \omega_\gamma},
\end{equation}
which leads to
 $w_\gamma \propto (S_T + \alpha(\gamma,S_T)I_p)^{-1} d$.
For each fixed $\gamma$, if the dimension $p$ of $S_T$ increases, then the total variance of $S_T$ also increases, which in turn leads that $\alpha(\gamma,S_T)$ in (\ref{eq:alpha-st}) be increasing. To lessen the technical difficulty in the exposition for this section, we use the ridge parameterization by $\alpha$ for the continuum directions.
In particular, we parameterize the continuum directions by $\alpha_p := \alpha  p $, which is an increasing function of the dimension $p$. For each $p$, we consider the set of sample continuum directions, denoted by $\hat{w}_\alpha \propto  (S_T + \alpha_p I_p)^{-1} d$, for $\alpha \neq 0$.

The population counterpart of the sample continuum directions is defined similarly. For $\mu = \mu_1 - \mu_2$,  $\Sigma_B = \mu \mu^T$, $\Sigma_W = (\Sigma_1 + \Sigma_2)/2$, and $\Sigma_T = \Sigma + \Sigma_B$, the population continuum directions are parameterized by $\alpha$, and are denoted by
 $\omega_\alpha \propto (\Sigma_T + \alpha_p I_p)^{-1} \mu$.
%
Assume the following:
\begin{enumerate}
\item[C1.] There exists a constant $\delta^2\ge 0$ such that $p^{-1}\|\mu\|^2 \to \delta^2$ as $p\to\infty$.
\item[C2.] $p^{-1}\mbox{tr}(\Sigma_1) \to \sigma_1^2$, $p^{-1}\mbox{tr}(\Sigma_2) \to \sigma_2^2$ as $p \to \infty$.
\item[C3.] The eigenvalues of $\Sigma_1$ (and $\Sigma_2$) are sufficiently concentrated, in the sense that
    $ [\mbox{tr}(\Sigma_i^2)]^2 / [\mbox{tr}(\Sigma_i)]^2 \to 0 $ as $p\to\infty$, for $i = 1,2$.
\end{enumerate}

The condition C1 has also appeared in, e.g., \cite{Hall2005,Qiao2009,Ahn2012}, and requires that the true mean difference grows as the dimension increases.
The conditions C2 and C3 include the covariance matrix models for both independent variables and mildly-spiked cases (i.e., few eigenvalues are moderately larger than the others), and were first appeared in \cite{Ahn2007}.
These conditions can be  generalized and the Gaussian assumption can be relaxed, as done in, e.g., \cite{Jung2009a,Jung2012}, to produce the equivalent results shown below. We keep it simple for brevity.

The asymptotic behavior of the sample continuum directions $\hat{w}_\alpha$, when $p\to\infty$, is investigated in two ways. We first show that $\hat{w}_\alpha$ is inconsistent, and has a non-negligible constant angular bias when compared to its population counterpart $\omega_\alpha$. Despite the bias, the CDA, the classification rule discussed in Section~\ref{sec:classification}, can perfectly classify new observations under certain conditions.

\begin{thm}
\label{thm:simple}
Under the setting in this section, including the conditions C1---C3, the following holds.

(i) The sample continuum directions are inconsistent with its population counterparts. In particular, for any $\alpha\neq 0$,
$$\mbox{Angle}(\omega_\alpha, \hat{w}_\alpha) \to \cos^{-1}\left(\frac{\delta^2}{\delta^2 + \sigma_1^2/n_1 + \sigma_2^2/n_2} \right)^{1/2},$$ in probability as $p \to\infty$.

(ii) The probability that CDA classifies a new observation correctly tends to 1 as $p \to\infty$ if $\delta^2 > \left|\sigma_1^2/n_1 - \sigma_2^2/n_2 \right|$.
\end{thm}

Both results in Theorem~\ref{thm:simple} depend on the quantity $\delta^2$ in the condition C1, which may be interpreted as a signal strength. When $\delta^2$ is large, the sample continuum direction is less biased, and $\mbox{Angle}(\omega_\alpha, \hat{w}_\alpha)$ is small. On the other hand, if $\delta = 0$,  then $\hat{w}_\alpha$ is strongly inconsistent with $\omega_\alpha$, and $\hat{w}_\alpha$ is asymptotically orthogonal to $\omega_\alpha$. The performance of CDA also depends on $\delta^2$. Consider the case where $\sigma_1 = \sigma_2$ and $n_1 = n_2$. Then CDA classification is perfect whenever $\delta^2$ is positive. On the other hand, if $\delta = 0$, then the classification is only as good as random guess. These observations are consistent with  \cite{Hall2005} and \cite{Qiao2009}, in which HDLSS asymptotic behaviors of the centroid rule, SVM and DWD are studied.

We conjecture that if the within-covariance matrix $\Sigma_W$ has a large first eigenvalue (that is, a large variance of the first principal component), then the sample continuum direction is less biased than in Theorem~\ref{thm:simple},  even under smaller size of signal $\delta^2$. This conjecture seems to be true, as shown in the simulation studies in Section~\ref{sec:sim}, but rigorously proving this conjecture has been challenging.

\section{Computations}\label{sec:computation}

\subsection{Numerical algorithm for the binary supervision case}\label{sec:computation-binary}
When $S_B$ is of rank 1, or when the supervision is binary, Theorem~\ref{prop:ridgesolutionthm} can be used to compute a discrete sequence of the first continuum directions.
 In particular, there is a corresponding $\gamma$ for each ridge parameter $\alpha \in (-\infty, -\lambda_1] \cup [0,\infty)$.
 Let $M > 0$ be a maximum value for evaluating $\alpha$. In our experience it is sufficient to choose $M = 10\lambda_1$, ten times larger than the largest eigenvalue of $S_T$.
Define $\alpha_{(k)} = \frac{k}{K}M$  and $\alpha^{(k)} = -(1+\epsilon)\lambda_1-\frac{K-k}{K}M$ for $k = 0, \ldots, K$ for some number $K$. The small number $\epsilon>0$ keeps the matrix $S_T+ \alpha^{(k)}I_p$ invertible and was chosen to $0.01$ for numerical stability. For each $\alpha = \alpha_{(k)}$ or $\alpha^{(k)}$, we get $w_{\gamma(\alpha)} = (S_T + \alpha I_p)^{-1} d$, where
$d$ satisfies $S_B = d d^T$ and
$$\gamma(\alpha) = \frac{\alpha}{w_{\gamma(\alpha)}^T S_T w_{\gamma(\alpha)} + \alpha}.$$
The sequence $\{w_{\gamma(\alpha)}: \alpha = \alpha_{(k)}, \alpha^{(k)}, k = 0,\ldots, K\}$ is augmented by the two extremes $w_{MD} ( \propto d) $ and $w_{PCA}$.

If $d$ is orthogonal to all eigenvectors corresponding to $\lambda_1$, then $\gamma$ does not tend to infinity even though $\alpha$ has reached  $-\lambda_1$. In such a case, the remaining sequence of directions is analytically computed using Proposition~\ref{prop:ridgesolutionexceptionthm} in Appendix.

%
%
%
%

\subsection{Numerical algorithm for the general case}\label{sec:computing-general}

In general cases where $\mbox{rank}(S_B) >1$, the connection to generalized ridge solutions in Theorem~\ref{prop:ridgesolutionthm} does not hold. Even with binary supervision, when a sequence of continuum directions $\{w_{(1)},\ldots, w_{(\kappa)}\}$ is desirable, the ridge parameter $\alpha(\gamma)$ is different for different $k$ in $w_{(k)}$, even when $\gamma$ is held fixed. Here, we propose a gradient descent algorithm to sequentially solve (\ref{eq:CDA}) for a given $\gamma$.

We first discuss a gradient descent algorithm for $w_{(1)}$. Since the only constraint is that the vector $w$ is of unit size, the unit sphere $S^{p-1} = \{w\in\Real^p : w^Tw = 1\}$ is the feasible space. To make the iterate confined in the feasible space we  update a candidate $w_0$ with $w_1 = (w_0 + c\nabla_{w_0})/\norm{w_0 + c\nabla_{w_0}}$, for a step size $c>0$, where the gradient vector is
$\nabla_{w} = \frac{S_B w}{w^TS_B w} + (\gamma - 1)  \frac{S_T w}{w^T S_T w}$.
To expedite convergence, $c$ is initially chosen to be large so that $w_1 \approx \nabla_{w_0}/ \norm{\nabla_{w_0}}$. If this choice of $c$ overshoots, i.e., $T_\gamma(w_1) < T_\gamma(w_0)$, then we immediately reduce $c$ to unity, so that the convergence to maximum is guaranteed, sacrificing fast rate of convergence. The iteration is stopped if $1-|w_1^Tw_0| < \eps$ or $|T_\gamma(w_1) - T_\gamma(w_0)| <\eps$ for a needed precision $\eps>0$. The step size $c$ can be reduced if needed, but setting $c \ge 1$ has ensured convergence with a precision level $\eps = 10^{-10}$ in our experience.

For the second and subsequent directions, suppose we have $w_{(1)},\ldots,w_{(k)}$ and are in search for  $w_{(k+1)}$.
The $S_T$-orthogonality and the unit size condition lead to the feasible space
$\mathcal{S} = \{w \in S^{p-1}: w^T S_T w_{(\ell)} =0, \ell = 1,\ldots, k\}$.
Since any $w \in \mathcal{S}$ is orthogonal to $z_{(\ell)} := S_T w_{(\ell)}$, $\ell = 1,\ldots, k$, the solution lies in the nullspace of $Z_k = [z_{(1)}, \ldots, z_{(k)}]$.
We use orthogonal projection matrix $P_k = I - Z_k (Z_k^TZ_k)^{-1}Z_k$ to project the variance-covariance matrices $S_T$ and $S_B$ onto the nullspace of $Z_k$, and obtain $S_T^{(k)} = P_k S_T P_k$ and $S_B^{(k)} = P_k S_B P_k$. The gradient descent algorithm discussed above for $w_{(1)}$ is now applied with $S_B^{(k)}$ and $S_T^{(k)}$ to update candidates of $w_{(k+1)}$, without the $S_T$-orthogonality constraint.

The following lemma justifies this iterative algorithm converges to the solution $w_{(k+1)}$.

\begin{lem}\label{lem:optimization}
\begin{enumerate}
\item[(i)] Let $x_i^* = P_k x_i$ be the projection of $x_i$ onto the nullspace of $Z_k$. Write $X^* = [x_1^*,\ldots,x_n^*]$. Then $S_T^{(k)} = n^{-1}X^*(X^*)^T$ and $S_B^{(k)} = n^{-1} (X^*Y^T)(X^*Y^T)^T$.
\item[(ii)] For  $w \in \mathcal{S}$, $T_\gamma(w)  = (w^T S_B^{(k)} w) (w^T S_T^{(k)} w)^{\gamma - 1} := T^{(k)}_\gamma(w) $.
\item[(iii)] The solution $w_{(k+1)}$ of the unconstrained optimization problem $max_w T_\gamma^{(k)}$ satisfies $w_{(k+1)}^T S_T w_\ell = 0 $ for $\ell = 1,\ldots,k$.
\end{enumerate}
\end{lem}

It can be seen from Lemma~\ref{lem:optimization}  that the optimization is performed with the part of data that is $S_T$-orthogonal to $Z_k$. While making the optimization simpler, we do not lose generality because the original criterion $T_\gamma$ has the same value as $T_\gamma^{(k)}$ for candidate $w$ in the feasible region (Lemma~\ref{lem:optimization}(ii)). This with the last result (iii) shows that our optimization procedure leads to (at least) a local maximum in the feasible region.

Note that the sequence $\{w_{(1)},\ldots, w_{(\kappa)}\}$ depends on the choice of $\gamma$. To obtain a spectrum of continuum directions, one needs to repeat the iterative algorithm for several choices of $\gamma>0$.

\subsection{Efficient computation when $p \gg n $}

For large $p$, directly working with $p\times p$ matrices $S_T$ and $S_B$ needs to be avoided. For such cases, utilizing the eigendecomposition of $S_T$ (or, equivalently, the singular value decomposition of $X$) provides efficient and fast computation for continuum directions.
Write $S_T = U\Lambda U^T$, where $U = [u_1,\ldots,u_m]$ spans the column space of $S_T$, for $m = \min(n-1,p)$.
Then the algorithms discussed in the previous sections can be applied to $\tilde{S}_T = U^T S_T U = \Lambda$ and $\tilde{S}_B = U^TS_B U$, in place of $S_T$ and $S_B$, to obtain $\tilde{w}_{(\ell)} \in \Real^m$.
The continuum directions are then $w_{(\ell)} = U\tilde{w}_{(\ell)}$.
If $m \ll p$, this requires much less computing time than working with $S_T$ and $S_B$ directly. The next lemma ensures that our solution is the maximizer of the criterion (\ref{eq:CDA}).

\begin{lem} \label{lem:range}
Any maximizer $w$ of (\ref{eq:CDA}) lies in  the column space of $S_T$.
\end{lem}

In the case of binary supervision, one needs to avoid the inversion of large $p\times p $ matrix $S_T +\alpha I_p$. The continuum  directions are obtained via only involving the inversion of $m \times m$ matrices: $(S_T +\alpha I_p)^{-} d =  U (\Lambda + \alpha I_p)^{-1} U^T d.$
In all of our experiments, involving moderately large data sets, where $\max(p,n)$ is tens of thousands and $\min(p,n)$ is hundreds, the computation takes only a few seconds at most, compared to several minutes needed for the method of \cite{Clemmensen:2011}.

\section{Simulation studies} \label{sec:sim}
We present two simulation studies to empirically reveal the underlying model under which the continuum directions are useful. We  numerically compare the performance of CDA, the linear classification followed by continuum dimension reduction, with several other classification methods, in binary or multi-category classification.

\subsection{Binary classification}\label{sec:sim_binary}

For binary classification, our method is compared with LDA (using the pseudoinverse),
the features annealed independence rule (FAIR) by \cite{Fan2008}, the distance weighted discrimination (DWD) by \cite{Marron2007} and the sparse discriminant analysis (SDA) by \cite{Clemmensen:2011}.

The setup for the simulation study is as follows. We assume two groups with mean $\mu_1 = 0$ and $\mu_2 = c_0 (1_{s},0_{p-s})^T$ for some constant $c_0$, where $1_{s}$ is the vector $(1,\ldots,1)^T$ of length $s$,and $0_{p-s} = (0,\ldots,0)^{T}$. We choose $s = 10$ or $p/2$, to examine both sparse and non-sparse models. The common covariance matrix is $\Sigma_{\rho} = (1-\rho) I_p + \rho 1_p 1_p^T$ for $\rho \in \{0,0.1,0.25,0.5\}$. This so-called compound symmetry model allows examination from independent to highly correlated settings. The scalar $c_0 = 3 ({1_{s}^T \Sigma_{\rho}^{-1} 1_{s}})^{-1/2}$ varies for different $(p,\rho)$ to keep the Mahalanobis distance between $\mu_1$ and $\mu_2$ equal to 3.

Training and testing data of size $n_1 = n_2 = 50$ are generated from normal distribution of dimension $p = 200,400$ and $800$. The parameter $\gamma$ of CDA is chosen by the 10-fold cross-validation. The number of features for FAIR, as well as the tuning parameters for SDA, were also chosen by 10-fold cross-validation. The mean and standard deviation of the  misclassification rates, based on 100 replications, are listed in Table~\ref{tab:sim}.

\begin{table}[tbp]
\centering
\begin{tabular}{ccccccc}
\multicolumn{7}{c}{Sparse model with $s = 10$} \\
$\rho$ & $p$ & CDA &  LDA  & FAIR & DWD & SDA\\
 \hline
\multirow{3}{*}{$0$}
& $200$ &14.32 (3.45) &29.59 (5.31) &8.90 (3.10)&13.88 (3.33) & 8.17 (2.92)\\
& $400$ &19.70 (4.07) &34.76 (5.33) &9.02 (3.23)&19.28 (4.10)& 8.57 (2.65) \\
& $800$ &24.90 (4.78) &39.80 (4.97) &9.80 (3.46)&24.14 (4.36)& 9.64 (5.76) \\
\hline
\multirow{3}{*}{$0.1$}
& $200$ &11.27 (3.56) &20.37 (4.65)&48.25 (7.09)&36.99 (7.07) & 4.90 (2.31)\\
& $400$ &9.87 (3.30)&26.97 (6.04)   &49.39 (5.09) &45.11 (5.11)& 5.12 (2.30)\\
& $800$ &12.94 (3.79)&36.24 (5.72) &50.32 (4.19)&48.65 (4.69)  & 5.82 (5.05)\\
\hline
\multirow{3}{*}{$0.25$}
& $200$ &5.90 (2.72) &13.38 (4.16)  &48.98 (5.30) &42.37 (5.58) & 1.86 (1.34)\\
& $400$ &3.88 (2.15) &19.93 (4.61)  &50.55 (5.20) &47.61 (5.20)& 1.71 (1.22)\\
& $800$ &5.67 (2.62) &31.14 (5.32)  &49.05 (4.71) &48.03 (4.74)& 2.39 (5.32)\\
\hline
\multirow{3}{*}{$0.5$}
& $200$ & 0.61   (0.94)     &      4.77  (2.51)    &         49.76  (5.19)     &    46.03 (4.73)      &   0.10  (0.30)   \\
& $400$ & 0.27   (0.49)     &      9.21  (3.73)    &         48.90  (4.66)     &    47.18 (4.52)      &   0.09  (0.32)   \\
& $800$ & 0.47   (0.73)     &     19.82  (5.13)    &         50.22  (4.83)     &    49.40 (4.83)      &   0.16  (0.75)   \\
\\
\multicolumn{7}{c}{Non-sparse model with $s = p/2$} \\
$\rho$ & $p$ & CDA &  LDA  & FAIR & DWD & SDA\\
 \hline
\multirow{3}{*}{$0$}
& $200$ &   14.66 (4.42)      &  29.30 (5.34)      &    14.40 (4.26)      &  13.60 (4.05)     &   22.05 (4.15)   \\
& $400$ &   19.36 (4.29)      &  34.83 (5.44)      &    19.51 (4.67)      &  18.64 (4.41)     &   30.80 (3.83)   \\
& $800$ &   24.71 (3.95)      &  40.38 (5.36)      &    25.40 (4.91)      &  24.05 (4.16)     &   36.78 (4.44)   \\
\hline
\multirow{3}{*}{$0.1$}
& $200$  &    6.45 (2.90)      &  20.91 (5.02)      &    47.65 (5.87)      &  36.49 (5.89)     &   20.19 (4.03)   \\
& $400$  &    9.47 (3.70)      &  27.82 (4.94)      &    48.82 (5.29)      &  44.33 (5.29)     &   29.79 (4.93)   \\
& $800$  &   13.11 (3.60)      &  36.42 (5.76)      &    50.21 (5.07)      &  48.29 (4.97)     &   35.12 (4.52)   \\
\hline
\multirow{3}{*}{$0.25$}
& $200$ &    2.25 (1.92)      &  13.36 (4.31)      &    48.94 (5.18)      &  42.01 (5.37)     &   15.83 (4.10)   \\
& $400$ &    2.95 (1.75)      &  20.61 (5.17)      &    50.47 (5.74)      &  47.23 (5.38)     &   24.65 (4.43)   \\
& $800$ &    5.34 (2.75)      &  30.43 (5.85)      &    50.24 (4.98)      &  49.03 (5.11)     &   31.68 (4.06)   \\
\hline
\multirow{3}{*}{$0.5$}
& $200$ &     0.56  (0.82)      &    5.60   (3.07)    &       49.91   (5.48)      &    45.69    (5.60)     &     7.31     (3.04)    \\
& $400$ &     0.24  (0.45)      &    9.68   (3.83)    &       49.45   (5.49)      &    47.32    (5.33)     &    16.69     (3.96)    \\
& $800$ &     0.39  (0.57)      &   20.93   (5.39)    &       49.84   (5.59)      &    49.05    (5.22)     &    26.02     (4.36)    \\
\end{tabular}
\caption{Performance of binary classification. Compound Symmetry model with high dimension, low sample size data: Mean misclassification error (in percent) with standard deviation in parentheses.}
\label{tab:sim}
\end{table}

Our results show that CDA performs much better than other methods when the  variables are strongly correlated ($\rho = 0.1, 0.25, 0.5$), for non-sparse models.
In the independent setting ($\rho = 0$), the performance of CDA is comparable to DWD.
FAIR is significantly better than CDA under sparse model with independent variables, because the crucial assumption of FAIR that the non-zero coordinates of $\mu_1 - \mu_0$ are sparse is also satisfied. However, FAIR severely suffers from the violation of the independence assumption, in which case their classification rates are close to 50$\%$. DWD also suffers from the highly correlated structure. SDA performs well for all settings under the sparse model, as expected. However, for non-sparse models, CDA performs significantly better than SDA.

Another observation is that the performance of LDA is better for a larger $\rho$. A possible explanation is that the underlying distribution $N(\mu_i,\Sigma)$ becomes degenerate as $\rho$ increases. The true covariance matrix has a very large first eigenvalue $\lambda_1 = p \rho + (1-\rho)$ compared to the rest of eigenvalues $\lambda_j = 1-\rho $, $2\le j \le p$. As conjectured in Section~\ref{sec:high-d}, both LDA and CDA benefit from extensively incorporating the covariance structure, in spite of the poor estimation of $\Sigma_\rho$ when $p \gg n$. Note that in terms of the conditions C1---C3 in Section~\ref{sec:high-d}, all of these models have signal strength $\delta^2 = 0$ and the condition C3 is violated when $\rho >0$.

Poor performance of FAIR for the strongly correlated case is also reported in \cite{fan2012road}, where they proposed the regularized optimal affine discriminant (ROAD), which is computed by a coordinate descent algorithm. Due to the heavy computational cost, we excluded the ROAD as well as the linear programming discriminant rule (LPD) by \cite{Cai2011}. We exclude results from \cite{Wu2009} since the performance of SDA \citep{Clemmensen:2011} were uniformly better than the method of Wu et al. These methods aim to select few features as well as to classify, based on assumptions of sparse signals. CDA does not require such assumptions.

\subsection{Multi-category classification}\label{sec:sim_multi}

For multi-category classification, CDA is compared with the reduced-rank LDA \citep[\emph{cf}.][]{Hastie2009} and SDA \citep{Clemmensen:2011}.

The setup in the simulation study is as follows. We assume $K = 3$ groups with means $\mu_1 = 0$, $\mu_2 = c_0 (1_{s},0_{p-s})^T$ and $\mu_3 = c_0 (0_{s},1_{s},0_{p-2s})^T$, for either $s=10$ or $s = p/2$.  The common covariance matrix  $\Sigma_{\rho}$ is the compound symmetry model, parameterized by $\rho \in \{0, 0.1, 0.25, 0.5\}$, and the scalar $c_0$ is set as explained in Section~\ref{sec:sim_binary}.

Training and testing data of size $n_1 = n_2 = n_3 = 50$ are generated from normal distribution of dimension $p = 200,400$ and $800$. The classification performances of CDA, reduced-rank LDA and SDA for these models are estimated by 100 replications, and are summarized in Table~\ref{tab:sim-MULTI}.

\begin{table}[tbp]
\centering
\begin{tabular}{ccccc}
\multicolumn{5}{c}{Sparse model with $s = 10$} \\
  & $p$ & CDA &  Reduced-rank LDA &  SDA\\
 \hline
\multirow{3}{*}{$\rho = 0$}
& $200$ &20.82 (4.61)&31.72 (5.71)&13.99 (3.92)\\
& $400$ &28.16 (4.96)&34.42 (5.22)&15.62 (5.90)\\
& $800$ &34.86 (5.31)&39.24 (5.35)&15.94 (5.77)\\
\hline
\multirow{3}{*}{$\rho = 0.1$}
& $200$ &14.01 (4.05)&22.94 (8.06) & 9.06 (2.83) \\
& $400$ &20.93 (5.85)&28.77 (12.03)&10.37 (4.76) \\
& $800$ &30.69 (9.01)&36.52 (12.61)&11.36 (5.78)  \\
\hline
\multirow{3}{*}{$\rho = 0.25$}
& $200$ & 6.38 (2.79)&15.39 (7.51)&3.71 (2.13) \\
& $400$ &12.60 (4.90)&25.62 (15.40)&4.06 (2.37) \\
& $800$ &20.80 (8.45)&30.06 (13.09)&3.88 (2.76)\\
\hline
\multirow{3}{*}{$\rho = 0.5$}
& $200$ & 0.89  (0.93)   &    13.67  (13.57)   &    0.28   (0.56)     \\
& $400$ & 1.21  (1.52)   &     4.42  (6.41)    &    0.34   (0.68)    \\
& $800$ & 4.54  (3.52)   &     8.87  (6.81)    &    0.33   (0.64)     \\
\\
\multicolumn{5}{c}{Non-sparse model with $s = p/2$} \\
  & $p$ & CDA &  Reduced-rank LDA &  SDA\\
 \hline
\multirow{3}{*}{$\rho = 0$}
& $200$      &    21.31 (4.40)       & 32.50  (5.45)     &  37.84 (4.91)  \\
& $400$      &    28.24 (4.73)       & 34.47  (5.15)     &  47.17 (5.18)  \\
& $800$      &    34.10 (5.41)       & 38.47  (5.51)     &  53.73 (4.48)  \\
\hline
\multirow{3}{*}{$\rho = 0.1$}
& $200$    &     5.27 (2.26)       & 47.37  (9.28)     &  30.87 (5.25)  \\
& $400$    &     9.83 (3.02)       & 52.58  (10.86)     &  38.43 (5.16)  \\
& $800$    &    23.70 (4.89)       & 38.79  (8.76)     &  44.66 (4.86)  \\
\hline
\multirow{3}{*}{$\rho = 0.25$}
& $200$ &     1.40 (1.45)       & 54.88  (10.55)     &  23.34 (5.47)  \\
& $400$ &     2.86 (1.71)       & 37.03  (9.61)     &  31.96 (5.12)  \\
& $800$ &     9.17 (3.04)       & 48.41  (13.27)     &  39.79 (5.04)  \\
\hline
\multirow{3}{*}{$\rho = 0.5$}
& $200$ &    0.06   (0.24)  &    34.19  (8.78)   &    11.07   (4.58)   \\
& $400$ &    0.10   (0.36)  &    45.28   (11.03)  &    21.79   (5.80)  \\
& $800$ &    0.51   (0.83)  &    30.81   (7.94)   &    32.86   (5.61)   \\
\end{tabular}

\caption{Performance of multi-category classification. Compound Symmetry model with high dimension, low sample size data: Mean misclassification error (in percent) with standard deviation in parentheses.}
\label{tab:sim-MULTI}
\end{table}

The simulation results for multi-category classification provide a similar insight obtained from the binary classification study. CDA performs better when the correlation between variables is strong for both sparse and non-sparse models. Our method is outperformed by SDA for the sparse model, but has significantly smaller misclassification rates for non-sparse models.

In summary, when the true mean difference is non-sparse and the variables are highly correlated, the proposed method performs better than competitors under high-dimension, low-sample-size situations
for both binary and multi-categoty classification problems.
When the variables are uncorrelated, we also checked that larger values of $c_0$ ensure good performance of the proposed method, as shown in Theorem~\ref{thm:simple}. Our method requires only a split second for computation, while SDA takes tens of seconds for the data in this study.

\section{Real data examples}\label{sec:realdata}
In this section, we provide three real data examples, where the supervision information is categorical with two or more categories.

\subsection{Leukemia data}
We first use the well-known data set of \cite{Golub1999}, which consists of expression levels of 7129 genes from 72 acute leukemia patients. The data are prepared as done in \cite{Cai2011}. In particular, 140 genes with extreme variances, i.e., either larger than $10^7$ or smaller than $10^{3}$ are filtered out. Then genes with the 3000 largest absolute $t$-statistics were chosen. The dataset included 38 training cases (27 AMLs and 11 ALLs) and 34 testing cases (20 AMLs and 14 ALLs).

\begin{figure}[tb]
\centering
\ifpdf
\vskip -2.5in
    \includegraphics[width=1\textwidth]{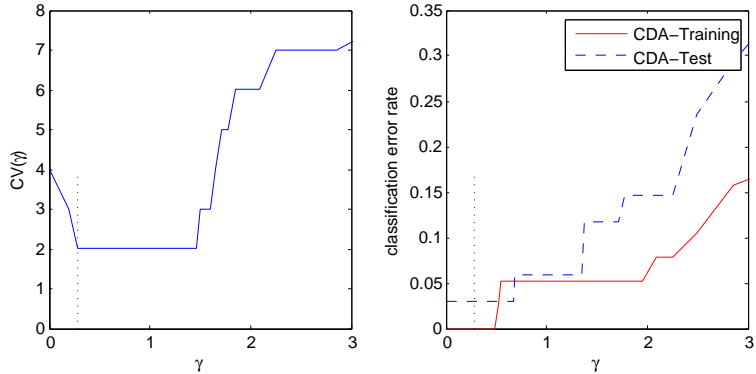}
\vskip -2.5in
\else
    \includegraphics[width=0.7\textwidth]{Leukimia_CDA.eps}
\fi
\caption{Left: Cross validatory errors for $\gamma \in [0,3]$ evaluated for Leukemia data. The $\hat\gamma = 0.279$ (located at the vertical dotted line) is the smallest $\gamma$ that minimizes $CV(\gamma)$. Right: Classification error rates of training and testing set for different $\gamma$s.}
\label{fig:LeukimiaCDA}
\end{figure}

With binary classification in mind, we obtain $w_\gamma$ for a discrete set of $0 \le \gamma < \infty$, using the computational procedure discussed in Section~\ref{sec:computation-binary}. A 10-fold cross-validation leads to $\hat\gamma = 0.279$. As shown in Fig.~\ref{fig:LeukimiaCDA}, the smallest cross validatory misclassification rate is $CV(\hat\gamma) = 2/38$. (We chose to use the smallest $\gamma$ among all minimizers of $CV(\gamma)$.)
Figure~\ref{fig:LeukimiaCDA} also shows the classification errors of training and testing data for different $\gamma$. For smaller $\gamma$ values, including $\gamma = 0$ (corresponding to MDP) and $\hat\gamma$, the classification errors are 1 out of 34 for the test set, and 0 out of 38 for the training set. In comparison, LDA, IR, DWD and SVM result in 2--6 testing errors. From the work of \cite{Fan2008} and \cite{Cai2011}, FAIR and LPD makes only 1/34 testing error. Sparse LDA methods, SLDA of \cite{Wu2009} and SDA of \cite{Clemmensen:2011}, also performed quite well. The results are summarized in Table~\ref{tab:leukimia}.

\begin{table}[tb]
\centering
\begin{tabular}{cccccccccc}
               & CDA & LDA   & IR   & DWD   & SVM & FAIR & LPD  &SLDA &SDA\\
\hline
Training error& 0/38 & 1/38 & 1/38  &  0/38 & 0/38 & 1/38 & 0/38 & 0/38 & 0/38\\
Testing error & 1/34 & 6/34 & 6/34  &  2/34 & 5/34 & 1/34 & 1/34 & 3/34 & 2/34
\end{tabular}
\caption{Classification error of Leukemia data.}
\label{tab:leukimia}
\end{table}

\subsection{Liver cell nuclei shapes}
In a biomedical study, it is of interest to quantify the difference between normal and cancerous cell nuclei, based on the shape of cells.
We analyze discretized cell outlines, aligned to each other to extract shape information \citep{Wang2011b}. 
The data consist of outlines from $n_1 = 250$ normal liver tissues and $n_2 = 250$ hepatoblastoma tissues. Each outline is represented by 90 planar landmarks, leading to $p = 180$.

In the context of discriminating the disease based on the cell shapes, we compare our method with LDA, DWD, FAIR, and a quadratic discriminant analysis (QDA). As explained in Section~\ref{sec:sim}, the threshold value of FAIR is chosen by cross validation. The QDA is modified to have smaller variability by using a ridge-type covariance estimator. 

For the comparison, we randomly assign 50 cases as a testing data set, and each classifier is calculated with the remaining 450 cases. The empirical misclassification rates of classifiers are computed based on the training dataset and on the testing dataset. This is repeated for 100 times to observe the variation of the misclassification rates. For the continuum directions with varying $\gamma$, we observe that the misclassification rates become stable as $\gamma$ increases, as shown in Fig.~\ref{fig:liverCDA}. Both the training and testing error rates become close to $1/3$ as $w_\gamma$ moves closer to MD and to PCA. This is because, for this dataset, $w_{MD}$ and $w_{PCA}$ are close to each other with $\mbox{angle}(w_{MD},w_{PCA}) = 6.67^\circ$, and both exhibit good classification performances, with error rate close to $1/3$. For each training dataset, $\hat\gamma$ is chosen by the cross validation. Many chosen $\hat\gamma$s have values between $(0.1,0.5)$, but a few of those are as large as $\gamma = 3$, as shown in Fig.~\ref{fig:liverCDA}.
The performance of CDA with cross-validated $\gamma$ is compared with other methods in Table~\ref{tab:liver}.
Based on the testing error rate, CDA performs comparable to more sophisticated methods such as FAIR and DWD. Both LDA and QDA tend to overfit and result in larger misclassification rates than other methods.

\begin{figure}[tb]
\centering
\ifpdf
\vskip -0.7in
    \includegraphics[width=0.46\textwidth]{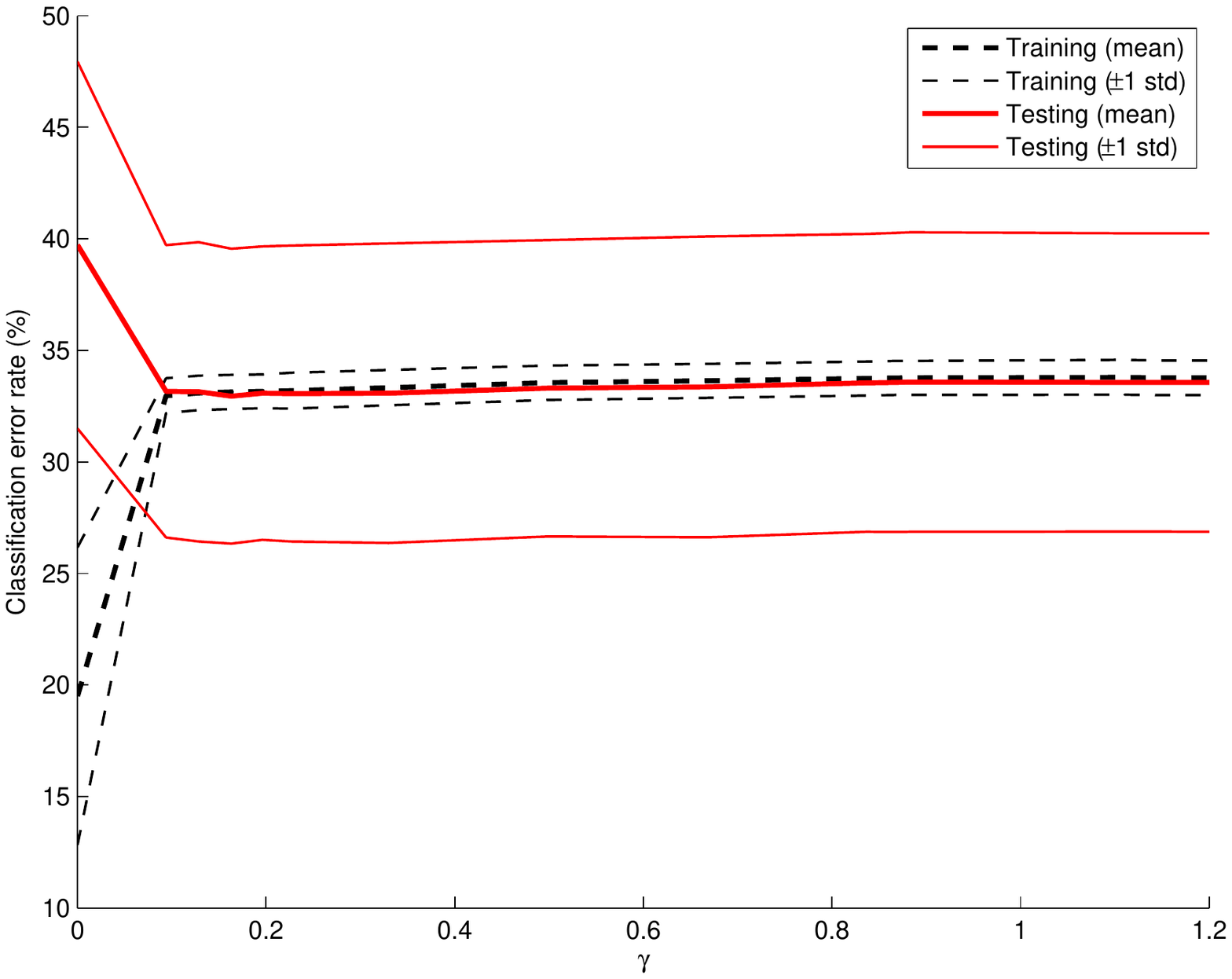}
    \includegraphics[width=0.46\textwidth]{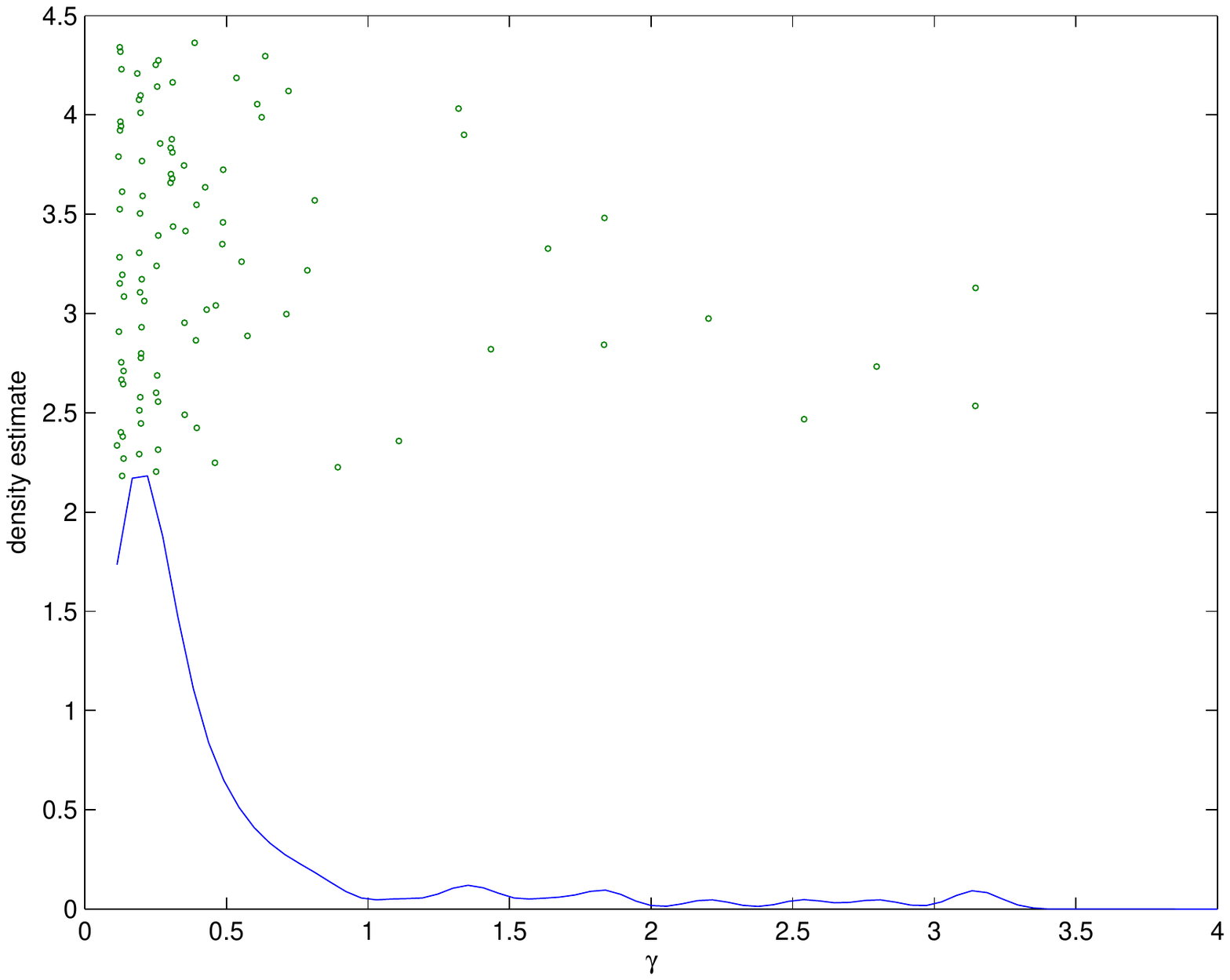}
\vskip -0.7in
\else
    \includegraphics[width=0.4\textwidth]{livercellCDAgamma.ps}
    \includegraphics[width=0.4\textwidth]{nuclei_gammaplot.ps}
\fi
\caption{Left: Classification error rates of training and testing set for different $\gamma$s. Right: A jitter plot with a density estimate for values of $\hat\gamma$ chosen by the cross validation.}
\label{fig:liverCDA}
\end{figure}
\begin{table}[tb]
\centering
\begin{tabular}{ccccccc}
& CDA & LDA & DWD & FAIR & QDA \\
\hline
Train &33.3 (0.79) & 13.9 (1.03) & 30.7 (0.78) & 32.6 (0.84) & 6.7 (2.85)\\
Test  &33.7 (6.38) & 37.4 (6.48) & 33.6 (6.33) & 33.3 (6.17) & 34.4 (6.85)
\end{tabular}
\caption{Misclassification rate (in percent) of liver nuclei outlines data. Mean and standard deviation of ten repetitions are reported.}
\label{tab:liver}
\end{table}


\subsection{Invasive lobula breast cancer data}
Invasive lobula carcinoma (ILC) is the second most prevalent subtype of invasive breast cancer. We use the protein expression data of $n = 817$ breast tumors, measured by RNA sequencing \citep{ciriello2015comprehensive}, to demonstrate the use of continuum directions when the supervision information is categorical with 5 possible values. The dataset consists of $p=16,615$ genes of $n = 817$ breast tumor samples, categorized into five subtypes---luminal A, basal-like, luminal B, HER2-enriched, and normal-like---by a pathology committee. Despite the large size of the data, computing the continuum directions is fast (few seconds, using a standard personal computer). Figure~\ref{fig:LobularFreezeCDAm} displays the spectrum of continuum dimension reduction, parameterized by the meta-parameter $\gamma >0$.

\begin{figure}[t]
\centering
\ifpdf
    \includegraphics[width=1\textwidth]{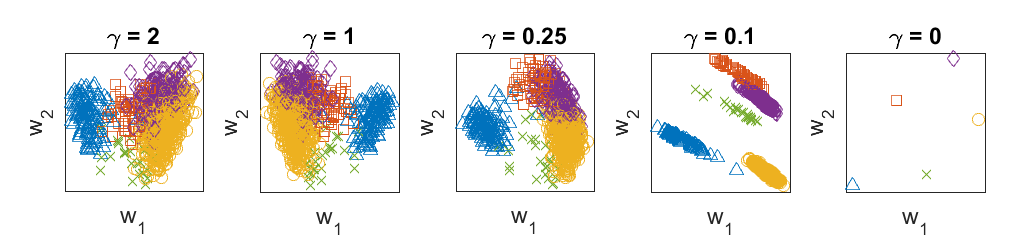}
\else
    \includegraphics[width=1\textwidth]{LobularFreezeCDAm.ps}
\fi
\caption{ILC data projected onto the first two continuum directions, for different choices of $\gamma$. Different colors represent different subtypes of ILC.}
\label{fig:LobularFreezeCDAm}
\end{figure}

To compare the performance of the multicategory classification with the reduced-rank LDA  and SDA of \cite{Clemmensen:2011}, we keep only the 500 genes with the largest standard deviations, and formed a training set of $409$ samples and a testing set of $408$ samples. For each of the classifiers, the training set is used to train the classification rule, while the testing set is used to estimate the misclassification error. We randomly permute the memberships to the training and testing sets, for 10 times.

The result of experiment is summarized in Table~\ref{tab:0}. Our method exhibits the lowest misclassification error rates. Poor performance of SDA may indicate that the true signal in the data is not sparse. As expected, the reduced-rank LDA severely overfits. 

\begin{table}[tb]
\centering
\begin{tabular}{cccc}
& CDA & Reduced-rank LDA &SDA \\
\hline
Train & 10.9 (3.42) &    0 (0)    & 9.58 (7.63)  \\
Test  & 14.5 (1.64) & 26.0 (1.91) & 28.6 (18.8) \\
\end{tabular}
\caption{Misclassification rates (in percent) of invasive lobula breast cancer data. Mean and standard deviation of ten repetitions are reported.}
\label{tab:0}
\end{table}

\section{Discussion}
We proposed a criterion evaluating useful multivariate direction vectors, called continuum directions, while the degrees of supervision from an auxiliary data set are controlled by a meta-parameter $\gamma$. An application of the proposed dimension reduction to classification was also discussed. Numerical properties of the proposed classifier have demonstrated good performance for high dimensional situation. In particular, our method outperforms several other methods when the variance of the first principal component is much larger than the rest.

The proposed method is akin to the continuum regression and connects several well-known approaches, LDA, MDP, MD, ridge estimators and PCA, thus providing a simple but unified framework in understanding the aforementioned methods.
There are several other criteria that also give a transition between LDA (or MDP) and PCA. A slightly modified criterion from (\ref{eq:CDAcriterion}),
$ F_\alpha(w) = (w^TS_Bw)^2 / |w^T(S_T + \alpha I_p) w |$ with the constraint $w^Tw = 1$, gives the ridge solution $\tilde{w}_\alpha = (S_T + \alpha I)^{-} d $ with the same $\alpha \in (-\infty, \lambda_1) \cup [0,\infty)$. This criterion is first introduced in a regression problem \citep{Bjorkstrom1999}, but has not been adopted into classification framework.  \cite{Wang2011} proposed a modified Fisher's criterion
\begin{equation}\label{eq:Wang}
\tau_\delta(w) = \frac{w^TS_Tw}{w^T(S_W+ \delta I)w},
\end{equation}
that bridges between LDA and PCA. For $\delta = 0$, the criterion (\ref{eq:Wang}) becomes identical to equation (\ref{eq:Fishercriterion}) up to the constant 1, thus equivalent to LDA.  In the limit of $\delta\to\infty$, $\delta \tau_\delta(w)$  converges to the criterion for $w_{PC1}$. The maximizer of $\tau_\delta$ is a solution of a generalized eigenvalue problem. We leave further investigation of these criteria as future research directions.


\cite{lee2013hdlss} also discussed discrimination methods that bridge MDP and MD, in high dimensions. The method of \cite{lee2013hdlss} is in fact equivalent to a part of continuum directions, restricted for $\gamma \in [0,1]$. In this paper, the continuum between MDP to PCA is completed by  also considering $\gamma > 1$, the method is extended for supervised dimension reduction, and a connection to continuum regression is made clear.

The study for HDLSS asymptotic behavior of the continuum directions has a room for more investigation.
We conjecture that the magnitude of large eigenvalues, in fast-diverging eigenvalue models, is a key parameter for successful dimension reduction, which may be shown using HDLSS asymptotic investigation similar to \cite{Jung2012}.

\section*{Acknowledgments}
The author is grateful to Gustavo Rohde for sharing the nuclei cell outlines data, and to Jeongyoun Ahn and Myung Hee Lee for their helpful suggestions and discussions.

 \appendix

\section{Technical details}

\subsection{Proof of Theorem~\ref{prop:ridgesolutionthm}.}

In a multivariate linear regression problem, with the $n \times p$ design matrix $X$ and the $n$ vector $y$ of responses,
denote a regressor by $w^T x $ a linear combination of $p$ variables.  Both $X$ and $y$ are assumed centered. Let $V(w) = w'X^TXw$ be the sample variance of the regressor. Let $K (w) = y^T Xw$ be the sample covariance between the regressor and $y$ and $R(w)$ be the sample correlation, which is proportional to $K/\sqrt{V}$.  The following theorem is from \cite{Bjorkstrom1999}.
\begin{thm}[Proposition 2.1 of \cite{Bjorkstrom1999}]\label{thm:Bjorkstrom}
If a regressor $w_f$ is defined according to the rule
$$ wf = \arg\max_{\norm{w} = 1} f (K^2(w),V(w)),$$
where $f(K^2,V)$ is increasing in $K^2$ (or $R^2$) for constant $V$, and increasing in $V$ for constant $R^2$, and if $X^Ty$ is not orthogonal to all eigenvectors corresponding to the largest eigenvalue $\lambda_1$ of $X^TX$, then there exists a number $\alpha$ such that $w_f \propto (X^TX + \alpha I)^{-1}$, including the limiting cases $\alpha \downarrow 0, \alpha \uparrow \infty$ and $\delta \uparrow -\lambda_1$.
\end{thm}

A two-group classification problem is understood as a special case of regression. In particular, let $y$ be $ +1$ if the $i$th observation is in the first group or $-1$ if it is in the second group. Then the total variance matrix $S_T \propto X^TX$ and  the mean difference $d = X^Ty$. The criterion (\ref{eq:CDAcriterion}) is $K^2(w) V^{\gamma-1}(w)$, which satisfies the assumptions of Theorem~\ref{thm:Bjorkstrom}. Theorem~\ref{prop:ridgesolutionthm} is thus a special case of Theorem \ref{thm:Bjorkstrom}.

\subsection{Analytic solution for the rare case}

The ridge solution may not give a global maximum of $T_\gamma$ when the assumption in Theorem~\ref{prop:ridgesolutionthm} does not hold. We give an analytic solution for such a case.
It is convenient to write $w$ in the canonical coordinates of $S_T$. Let $S_T = U\Lambda U^T$ be the eigen-decomposition of $S_T$ with $\Lambda = \mbox{diag}(\lambda_1,\ldots, \lambda_m)$, $U = [u_1,\ldots, u_m]$ for $m = \min(n-1,p)$, with convention $\lambda_i \ge \lambda_{i+1}$. To incorporate any duplicity of the first eigenvalue let $\iota$ represent the number of eigenvalues having the same value as $\lambda_1$, that is, $\lambda_1 = \ldots = \lambda_\iota$. Denote $\Lambda_1 =  \mbox{diag}(\lambda_1,\ldots, \lambda_\iota) = \lambda_1 I_\iota$ and $\Lambda_2 =  \mbox{diag}(\lambda_{\iota+1},\ldots, \lambda_m)$.
Let $z = U^T w$ and $\delta = (\delta_1,\ldots, \delta_m)^T =  U^T d$. 

\begin{prop}\label{prop:ridgesolutionexceptionthm}
Suppose $d$ is orthogonal to all eigenvectors corresponding to $\lambda_1$ and is not orthogonal to all eigenvectors corresponding to $\lambda_{\iota+1}$.
Let
$$ z_\alpha =
  \frac{(\Lambda_2 + \alpha I)^{-1} \delta_2 }{\sqrt{\delta_2^T (\Lambda_2 + \alpha I)^{-2} \delta_2} }  \ \  \mbox{for} \ \  \alpha \in (-\infty, -\lambda_1] \cup [0,\infty).
$$
\begin{enumerate}
\item[(i)] If $ z_{-\lambda_1}^T (\lambda_1 I - \Lambda_2) z_{-\lambda_1} \le \lambda_1/\gamma$, then $w_\gamma = U \tilde{z}$, $ \tilde{z}^T = [0_\iota^T, z_\alpha^T]$ for some $\alpha \in (-\infty, -\lambda_1] \cup [0,\infty)$.
\item[(ii)] If $ z_{-\lambda_1}^T (\lambda_1 I - \Lambda_2) z_{-\lambda_1} > \lambda_1/\gamma$, then there exist multiple solutions $w_\gamma = U \hat{z}$, $\hat{z}^T = (\hat{z}_1^T,\hat{z}_2^T)$, of (\ref{eq:CDAcriterion}) satisfying
$$\hat{z}_1 \in \{z_1 \in  \Real^{\iota} : z_1^Tz_1 = 1 - {\lambda_1} / ({\gamma}{z_{-\lambda_1}^T (\lambda_1 I - \Lambda_2) z_{-\lambda_1}})  \}$$ and $$\hat{z}_2 = \sqrt{\frac{\lambda_1}{\gamma}} \frac{(\Lambda_2-\lambda_1 I)^{-1} \delta_2}{\sqrt{ \delta_2^T (\lambda_1 I - \Lambda_2)^{-1} \delta_2} }.$$
\end{enumerate}
\end{prop}

\begin{proof}[Proof of Proposition \ref{prop:ridgesolutionexceptionthm}]

Recall $z = U^T w$ and $\delta = (\delta_1,\ldots, \delta_m)^T =  U^T d$. Grouping $z$ and $\delta$ into the first $\iota$ elements and the rest, write $z^T = (z_1^T, z_2^T)$, $\delta^T = (\delta_1^T, \delta_2^T)$. If $d$ is orthogonal to all eigenvectors corresponding to $\lambda_1$, then $\delta_1 = 0$. Rewriting equation (\ref{eq:Lagrangian}) in the eigen-coordinates gives two systems of equations
\begin{eqnarray}
       0 + (\gamma - 1) \frac{\Lambda_1 z_1}{z_1^T\Lambda_1z_1+z_2^T\Lambda_2z_2} - \gamma z_1 &=& 0, \label{eq:orthopart1}\\
       \frac{\delta_2}{z_2^T\delta_2} + (\gamma -1) \frac{\Lambda_2 z_2}{z_1^T\Lambda_1z_1+z_2^T\Lambda_2z_2} - \gamma z_2&=& 0. \label{eq:orthopart2}
\end{eqnarray}
If $\norm{z_1} > 0 $, then we have from (\ref{eq:orthopart1})
\begin{equation}\label{eq:orthopart3}
\lambda_1 \norm{z_1}^2 = z_1^T\Lambda_1 z_1 = \frac{\gamma - 1}{\gamma} \lambda_1 - z_2^T\Lambda_2 z_2.
\end{equation}
Equations (\ref{eq:orthopart2}) and (\ref{eq:orthopart3}) lead to
 $$ z_2 = c_\gamma (\Lambda_2 - \lambda_1 I)^{-1} \delta_2,$$
where $c_\gamma$ satisfies $c_\gamma^2 = -{\lambda_1} / ( {\gamma} {\delta_2^T (\Lambda_2 -\lambda_1 I)^{-1} \delta_2}) $, which is obtained  from the constraint $\norm{z_1}^2 + \norm{z_2}^2 = 1$. Finally, we check that such a solution exists if $z_2^Tz_2 \le 1$, that is,
\begin{equation}\label{eq:orthopartexistence}
\frac{\sum_{i=\iota+1}^m \lambda_1\delta_i^2/(\lambda_1- \lambda_i)^2}{\sum_{i=\iota+1}^m \delta_i^2/(\lambda_1- \lambda_i)} \le \gamma.
\end{equation}

The criterion $T_\gamma$ in the canonical coordinate is proportional to
$$T_\gamma (z) = (z_2^T\delta_2)^2 (\lambda_1 z_1^T z_1 + z_2^T \Lambda_2 z_2)^{\gamma-1}.$$
Thus $T_\gamma$ is maximized by $\hat{z}^T = (\hat{z}_1^T,\hat{z}_1^T)$ for any $\hat{z}_2 = \pm z_2$ and any $\hat{z}_1$ that  satisfies (\ref{eq:orthopart3}). This proves (ii).

If (\ref{eq:orthopartexistence}) does not hold, then by contradiction we have $\norm{z_1} = 0$. Thus $\tilde{z}$ is of the form $(0_\iota, z_2)$ for $z_2$ satisfying (\ref{eq:orthopart2}). Since the first coordinate of $\delta_{2}$ is nonzero, an application of Theorem \ref{prop:ridgesolutionthm} leads that there exists $\alpha \in (-\infty, -\lambda_{\iota+1}) \cup [0,\infty)$ such that $z_2 \propto  (\Lambda_2 + \alpha I)^{-} \delta_2$.

To conclude (i), we need to rule out the possibility of $\alpha$ having values in $(-\lambda_1, -\lambda_{\iota+1})$.
Let $M_k = M_k(a) = \delta_2^T (a I - \Lambda_2)^{-k} \delta_2$ for $k = 1,2,\ldots $.  The derivative of $M_k$ with respect to $a$ is $M_k^\prime = -k M_{k+1}$. We have $M_k(a) > 0$ for $a \in (\lambda_{\iota+1}, \lambda_{1}]$. The assumption of (i) is written as $\gamma \le \lambda_1 M_2(\lambda_1) / M_1(\lambda_1)$. It can be shown that $a M_2(a) / M_1(a)$ is a decreasing function of $a >  \lambda_{\iota+1}$. This leads to
\begin{equation}\label{eq:finaleq}
\gamma \le aM_2 / M_1, \ \ \mbox{ for any } \ a \in  (\lambda_{\iota+1}, \lambda_1].
\end{equation}
For $z_\alpha = (\Lambda_2 + \alpha I)^{-} \delta_2 / \norm{(\Lambda_2 + \alpha I)^{-} \delta_2} $,
$T_\gamma( (0_\iota, z_{-a}) ) = M_1^2 / M_2 (a - M_1/M_2)^{\gamma-1}$, and the derivative of $\log (T_\gamma)$
$$\frac{ 2( M_2^2 - M_1M_3)}{M_1M_2 (M_2 a - M_1) } ( \gamma M_1 - M_2 a) \ge 0 \  \mbox{ for any } \ a \in  (\lambda_{\iota+1}, \lambda_1].$$
We have used (\ref{eq:finaleq}) and the Cauchy-Schwartz inequality. Since $T_\gamma$ is increasing in $a$, any $z_{\alpha}$ with $\alpha \in (-\lambda_1, -\lambda_{\iota+1})$ can not be a  maximizer of $T_\gamma$ for any $\gamma$, which completes the proof.
\end{proof}

%


\subsection{Proofs of Proposition \ref{prop:RRtoMDP} and  Lemmas~\ref{lem:optimization}-\ref{lem:range}}

\begin{proof}[Proof of Proposition \ref{prop:RRtoMDP}]
We first show that $(S_T+\alpha I)^{-1} d \propto (S_W+\alpha I)^{-1} d $. Let $\Omega = S_W+\alpha I$, whose inverse exists for $\alpha > 0$. Then $S_T+\alpha I = \Omega + c_0 d d^T$ for $c_0 = \frac{n_1n_2}{n^2}$. By Woodbury's formula, $(S_T+\alpha I)^{-1} = \Omega^{-1} - c_1 \Omega^{-1}d d^T \Omega^{-1}$ for some constant $c_1$. Therefore,  $(S_T+\alpha I)^{-1} d = \Omega^{-1} d - c_1 \Omega^{-1}d d^T \Omega^{-1} d = c_2 \Omega^{-1} d \propto (S_W+\alpha I)^{-1} d$.

The ridge solution $w^{R}_\alpha $ lies in the range of $S_T$, as shown in Lemma~\ref{lem:range} in the Appendix. Writing $w^{R}_\alpha$ in the eigen-coordinates of $S_T$ makes the proof simple. Let $S_T = U\Lambda U^T$ be the eigen-decomposition of $S_T$ with $\Lambda = \mbox{diag}(\lambda_1,\ldots, \lambda_m)$, $U = [u_1,\ldots, u_m]$ for $m = \min(n-1,p)$. Then
for $z^{R}_\alpha = U^T w^{R}_\alpha$ and $\delta = (\delta_1,\ldots, \delta_m)^T =  U^T d$, we have
$z^{R}_\alpha \propto (\frac{\delta_1}{\lambda_1+\alpha},\ldots, \frac{\delta_m}{\lambda_m+\alpha})^T$, which leads to the continuity of $w^{R}_\alpha = U z^{R}_\alpha$ with respect to $\alpha \in [0,\infty)$. It is now easy to see that $w^{R}_{\alpha} \to w_{MDP} \propto U \Lambda^{-1} \delta$ as $\alpha \to 0$.
For the last argument, $z^R_\alpha \propto \alpha(\frac{\delta_1}{\lambda_1+\alpha},\ldots, \frac{\delta_m}{\lambda_m+\alpha})^T \to \delta$ as $\alpha \to \infty$.
\end{proof}

\begin{proof}[Proof of Lemma~\ref{lem:optimization}]
Part (i) is trivial. For part (ii), note that for all $w \in \mathcal{S}$, $P_k w = w$. Replacing $w$ by $P_k w$ in $T_\gamma(w)$ gives the result. For part (iii), we use Lemma~\ref{lem:range} in the Appendix which shows that the solution $w$ of maximizing $T^{(k)}_\gamma$ lies in the column space of $P_kS_TP_k$. Thus, the solution $w_{(k+1)}$ satisfies the constraint $w_{(k+1)}^T S_T w_\ell = 0 $ for $\ell = 1,\ldots,k$.
\end{proof}

\begin{proof}[Proof of Lemma~\ref{lem:range}]
Denote the column space of $S_T$ by $\mathcal{R}_T$. Let $\mbox{rank}(S_T) = m \le \min (n-1, p)$.  Then for any $w \in \Real^p$ with $\norm{w} = 1$, let $w_P$ be the orthogonal projection of $w$ onto $\mathcal{R}_T$. Then $\norm{w_P} \le 1$ where the equality holds if and only if $w \in \mathcal{R}_T$. Let $\tilde{w} = w_P/ \norm{w_P}$. Then since $w^TS_Tw = w_P^TS_Tw_P$ and $w^TS_Bw = w_P^TS_Bw_P$, we have for $\gamma \ge 0 $,
$$T_\gamma(w) = (\tilde{w}^TS_B \tilde{w})(\tilde{w}^TS_B \tilde{w})^{\gamma-1} \norm{w_P}^{2\gamma} \le T_\gamma(\tilde{w}).$$
Thus the maximizer of $T_\gamma(w)$ always lies in $\mathcal{R}_T$.
\end{proof}

\subsection{Proof of Theorem \ref{thm:simple}}

We first show that the true continuum direction is asymptotically parallel to the mean difference direction.
Assume without loss of generality that the true pooled covariance matrix $\Sigma_W$ is a diagonal matrix, for every $p$.

\begin{lem}
\label{thm:true}
Assume conditions C1---C3. For each $\alpha\neq 0$, $\mbox{Angle}(\omega_\alpha, \mu) \to 0$ as $p \to\infty$.

\end{lem}

\begin{proof}
[Proof of Lemma~\ref{thm:true}]

Let $A_p$ denote the $p\times p$ diagonal matrix with diagonal values $\lambda_i + \alpha_p$ where $\lambda_i$ is the $i$th largest eigenvalue of $\Sigma_W$. Using Woodbury's formula, we get
\begin{align*}
\omega_\alpha  \propto [A_p + \mu\mu^T]^{-1}\mu
               = A_p^{-1}\mu - \frac{A_p^{-1}\mu (\mu^TA_p^{-1}\mu)}{1+\mu^TA_p^{-1}\mu}
               \propto A_p^{-1}\mu.
\end{align*}
Then $\mbox{Angle}(\omega_\alpha, \mu) = \mbox{Angle}(A_p^{-1}\mu, \mu) = \cos^{-1}[\mu^T A_p^{-1}\mu / ( \norm{A_p^{-1} \mu} \norm{ \mu})  ]$. We then have
$\mu^T A_p^{-1}\mu \le (\lambda_p + \alpha_p)^{-1}\sum_{i=1}^n \mu^2_i  = (\lambda_p p^{-1} + \alpha)^{-1} \norm{\mu}^2 / p \to \delta^2/\alpha$,  $p^{1/2} \norm{A_p^{-1}\mu} \ge (\lambda_1 p^{-2} + \alpha)^{-1} p^{-1/2}\norm{\mu} \to \delta / \alpha$, as $p\to\infty$. This, together with the condition C1, leads that $\mbox{Angle}(A_p^{-1}\mu, \mu) \to 0$   as $p\to\infty$.
\end{proof}

We utilize a few relevant results in literature. Recall that $d = \bar{x}_1 - \bar{x}_2$ and $\mu = \mu_1-\mu_2$ are the sample and population mean difference vectors. The notation $\mbox{Angle}(x, \mathcal{R}_W)$, for $x \in \Re^p$, and a subspace $\mathcal{R}_W \subset \Re^p$, stands for the canonical angle, i.e. $\mbox{Angle}(x, \mathcal{R}_W) = \min_{y \in \mathcal{R}_W, y\neq 0} \mbox{Angle}(x,y)$.

\begin{lem}\label{lem:previous_result}
Assume the condition of Theorem~\ref{thm:simple}.

\begin{enumerate}
\item[(i)] \citep[][Theorem 3.]{Qiao2009} $p^{-1} \norm{d}^2 \to \delta^2 + \sigma_1^2/n_1 + \sigma_2^2/n_2$

\item[(ii)] \citep[][Theorem 6.]{Qiao2009} $\cos[\mbox{Angle}(d, \mu)] \to \left(\frac{\delta^2}{\delta^2 + \sigma_1^2/n_1 + \sigma_2^2/n_2} \right)^{1/2}$ in probability as $p \to \infty$.

\item[(iii)] \citep[][Theorem 1.]{Hall2005} If $\delta^2 > |\sigma_1^2 /n_1 - \sigma_2^2 / n_2|$, then the probability that a new datum from either $N(\mu_1,\Sigma_1)$ or $N(\mu_2,\Sigma_2)$ population is correctly classified by the centroid discrimination rule converges to 1 as $p\to\infty$. Here, the centroid discrimination rule classifies a new observation $x$ to the first group, if $\|x -\bar{x}_1\| < \|x -\bar{x}_2\| $.

\item[(iv)] \citep[][Theorem 1.]{Jung2009a} Each of $n_1+n_2-2$ nonzero eigenvalues of $p^{-1}S_W$ converges to either $\sigma_1^2$ or $\sigma_2^2$ in probability as $p \to\infty$.

\item[(v)] $\mbox{Angle}(d, \mbox{range}(S_W)) \to \pi/2$ in probability as $p \to \infty$.
\end{enumerate}
\end{lem}

\begin{proof}
[Proof of Lemma~\ref{lem:previous_result}]
The statements (i)-(iv) are modified from the original statements of the referenced theorems, and easily justified.

A proof of (v) is obtained by the following two facts. First, the column space of $S_W$ is spanned by $\{x_{ij}-\bar{x}_i\}$. Second, for each $(i,j)$, $\mbox{Angle}(d,x_{ij}-\bar{x}_i) \to 0 $ in probability as $p\to\infty$. The second result is obtained from the facts $p^{-1}\|x_{11} - \bar{x}_1\|^2 \to \sigma_1^2(n-1)/n$, and $p^{-1} d^T(x_{11} - \bar{x}_1) \to 0$ in probability as $p\to\infty$, as well as Lemma~\ref{lem:previous_result}(ii).
\end{proof}

Write the eigendecomposition of $S_W$ by $S_W = \widehat{U}_1\widehat{\Lambda}_W\widehat{U}_1^T$, where $\widehat{U}_1$ collects the $(n_1+n_2-2)$-dimensional eigenspace, corresponding to nonzero eigenvalues. Let $\widehat{U}_2$ denote the orthogonal basis matrix for the nullspace of $S_W$. Then $\widehat{U} = [\widehat{U}_1,\widehat{U}_2]$ is the $p\times p$ orthogonal matrix, satisfying $\widehat{U}\widehat{U}^T = \widehat{U}^T\widehat{U} = I_p$. Write  $d_1 = \widehat{U}_1^T d$, $d_2 = \widehat{U}_2^T d$ and $N = n_1+n_2-2$. Then, we can write
\begin{equation}\label{eq:decomposition}
\hat{w}_\alpha^R
\propto \widehat{U}_1 (\widehat{\Lambda}_W + \alpha_p I_{N})^{-1} d_1
 + \alpha_p^{-1}\widehat{U}_2 d_2 := b_\alpha.
\end{equation}

The following intermediate result concerning (\ref{eq:decomposition}) will be handy.
\begin{lem}\label{lem:new_lem}
Assume the condition of Theorem~\ref{thm:simple}.

(i) $p^{-1}\norm{d_1}^2 \to 0$, and $p^{-1}\norm{d_2}^2 \to \delta^2 + \sigma_1^2/n_1 + \sigma_2^2/n_2$  in probability as $p\to\infty$

(ii) $\mbox{Angle}(b_\alpha,d) \to 0 $ in probability as $p\to\infty$.
\end{lem}

\begin{proof}[Proof of Lemma~\ref{lem:new_lem}]
In this proof, every convergence is a convergence in probability as $p\to\infty$.

For a proof of (i), by Lemma~\ref{lem:previous_result}(i), showing $p^{-1}\norm{d_1}^2 \to 0$ is enough. From Lemma~\ref{lem:previous_result}(v), we have $\norm{\widehat{U}_1'd}/ \norm{d} = \cos(\mbox{Angle}(d, \mbox{range}(S_W)) \to 0$. Then
$p^{-1/2}\norm{d_1} = p^{-1/2}\norm{\widehat{U}_1'd} = p^{-1/2} \norm{d} (\norm{\widehat{U}_1'd}/ \norm{d})$, which converges to 0
since $ p^{-1/2} \norm{d}$ is stochastically bounded.

For (ii), we will show that
$| d^T b_\alpha| / \norm{b_\alpha}{\norm{d}} \to 1$.
From (\ref{eq:decomposition}), we have
\begin{equation}\label{eq:decomposition2}
p\norm{b_\alpha}^2 =\norm{ ( p^{-1}\widehat{\Lambda}_W + \alpha I_{N})^{-1} \frac{d_1}{\sqrt{p}} }^2 + \frac{\norm{d_2}^2}{\alpha^{2}p}.
\end{equation}
By Lemma~\ref{lem:previous_result}(iv), each element in the $N\times N$ matrix $(p^{-1}\widehat{\Lambda}_W + \alpha I_{N})$ converges to either $\sigma_1^2+\alpha$ or $\sigma_2^2+\alpha$. This fact and the part (i) shown above lead that the first term of (\ref{eq:decomposition2}) converges to 0. Therefore we have
\begin{equation}\label{eq:decomposition-int1}
p^{1/2}\norm{b_\alpha} \to \alpha^{-1} (\delta^2 + \sigma^2/n_1 + \tau^2/n_2)^{1/2}.
\end{equation}

Similarly, using the decomposition (\ref{eq:decomposition}), and Lemma~\ref{lem:previous_result}(iv) and the part (i) of Lemma~\ref{lem:new_lem}, we have
\begin{equation}\label{eq:decomposition-int2}
|d^Tb_\alpha| = p^{-1} d_1^T ( p^{-1}\widehat{\Lambda}_W + \alpha I_{N})^{-1} d_1 + \alpha^{-1}p^{-1}\norm{d_2}^2 \to \alpha^{-1} (\delta^2 + \sigma^2/n_1 + \tau^2/n_2) .
\end{equation}
Combining (\ref{eq:decomposition-int1}), (\ref{eq:decomposition-int2}) and Lemma~\ref{lem:previous_result}(i), we get
$$
\frac{| d^T b_\alpha| }{ \norm{b_\alpha}{\norm{d}}}
= \frac{| d^T b_\alpha| }{ (p^{1/2}\norm{b_\alpha} )(p^{-1/2}{\norm{d}})} \to 1,$$
as desired.
\end{proof}

We are now ready to prove Theorem~\ref{thm:simple}.

\begin{proof}[Proof of Theorem \ref{thm:simple}]
To show (i), it is enough to combine the results from Lemma~\ref{thm:true}, Lemma~\ref{lem:previous_result}(ii) and Lemma~\ref{lem:new_lem}(ii), which describes the asymptotic angles between the pairs $(\omega_\alpha,\mu)$, $(\mu,d)$, and $(d,\hat{w}_\alpha)$, respectively.

The statement (ii) is obtained by Lemma~\ref{lem:new_lem}(ii) and Lemma~\ref{lem:previous_result}(iii).
\end{proof}

\section*{References}

 \bibliographystyle{elsarticle-harv}
\bibliography{library}




\end{document}